\documentclass[11pt, tightenlines, notitlepage, aps, pra, longbibliography, superscriptaddress]{revtex4-1}
\usepackage[pdfpagelabels,pdftex,bookmarks,breaklinks]{hyperref}
\usepackage{color}
\usepackage{amsmath,amsfonts,amsthm,amssymb}
\usepackage{bbm} 
\usepackage{graphicx}

\usepackage{tikz}

\providecommand{\ket}[1]{\vert #1 \rangle}
\providecommand{\bra}[1]{\langle #1 \vert}

\newenvironment{qcircuit}{\begin{tikzpicture}[qcircuit] }{\end{tikzpicture}}

\def\circleH{.16} 
\def\targetH{.75*\circleH} 
\def\controlH{.375*\circleH} 
\def\vunit{4*\circleH} 

\tikzset{fc/.style={path picture={
	\filldraw[fill=white] (0,0) circle (\circleH);
	}}} 
\tikzset{fcb/.style={path picture={
	\filldraw[fill=black] (0,0) circle (\controlH);
	}}} 
\tikzset{plus/.style={path picture={
	\draw (0,0) circle (\targetH);
	\filldraw[fill=black] (-\targetH,0) -- (\targetH,0) (0,-\targetH) -- (0,\targetH) ;
	}}} 
\tikzstyle{gate} = [draw,fill=white,minimum size=1.5em,inner sep=2pt] 
\tikzstyle{ket} = [fill=white,minimum size=1.5em]
\tikzset{qcircuit/.style={thick, minimum size=3ex}} 

\newcommand{\fcb}{\qw node[fcb] {}} 
\newcommand{\plus}{\qw node[plus] {}} 
\renewcommand{\d}{-- +(0,-\vunit) +(0,0)}
\newcommand{\row}[1]{\draw (0,-#1*\vunit) }
\newcommand{\qw}{-- ++(1,0) }
\newcommand{\gate}[1]{\qw node[gate] {${#1}$}}

\newcommand{\qin}{\qw node[ket] {$\hphantom{\ket{0}}$}}
\newcommand{\qout}{node[ket] {$\hphantom{\bra{0}}$} }

\newcommand{\qket}{\qw node[ket] {$\ket{0}$}}
\newcommand{\qbra}{ node[ket] {$\bra{0}$} }

\definecolor{deeppurple}{RGB}{100,0,120} 
\definecolor{darkgreen}{RGB}{0,150,0}
\definecolor{darkblue}{RGB}{0,0,130}
\definecolor{magenta}{RGB}{200,0,200}
\hypersetup{colorlinks=true, linkcolor=deeppurple, citecolor=darkgreen, filecolor=red, urlcolor=darkblue}
\hypersetup{pdfauthor={Dave Bacon, Steven T. Flammia, Aram W. Harrow, Jonathan Shi}}
\hypersetup{pdftitle={Sparse Quantum Codes from Quantum Circuits}}

\newtheorem{theorem}{Theorem}
\newtheorem{lemma}[theorem]{Lemma}
\newtheorem{definition}[theorem]{Definition}

\renewenvironment{proof}[1][Proof]{\noindent\textbf{#1:} }{\ $\Box$\medskip}

\theoremstyle{remark}
\newtheorem*{remark}{Remark}

\theoremstyle{definition}

\providecommand*{\ket}[1]{|{#1}\rangle}
\providecommand*{\bra}[1]{\langle{#1}|}
\newcommand*{\ketbra}[2]{|{#1}\rangle\!\langle{#2}|}

\newcommand*{\proj}[1]{\ketbra{#1}{#1}}

\newcommand{\cA}{\mathcal{A}}
\newcommand{\cB}{\mathcal{B}}

\newcommand{\cG}{\mathcal{G}}
\newcommand{\cH}{\mathcal{H}}

\newcommand{\cK}{\mathcal{K}}
\newcommand{\cL}{\mathcal{L}}
\newcommand{\cM}{\mathcal{M}}
\newcommand{\cP}{\mathcal{P}}
\newcommand{\cS}{\mathcal{S}}

\newcommand*{\eps}{\epsilon}

\newcommand*{\ED}{\mathrm{ED}}
\DeclareMathOperator{\im}{im}

\DeclareMathOperator{\spack}{spackle}

\newcommand*{\vd}{{\vphantom{\dag}}}

\newcommand{\ot}{\otimes}

\newcommand{\comment}[1]{}

\newcommand{\init}{\text{in}}
\newcommand{\out}{\text{out}}

\def\be#1\ee{\begin{equation}#1\end{equation}}
\def\ba#1\ea{\begin{align}#1\end{align}}
\def\bas#1\eas{\begin{align*}#1\end{align*}}
\def\bit{\begin{itemize}}
\def\eit{\end{itemize}}
\def\benum{\begin{enumerate}}
\def\eenum{\end{enumerate}}
\newcommand{\fig}[1]{Fig.~\ref{fig:#1}}
\newcommand{\eq}[1]{(\ref{eq:#1})}
\newcommand{\secref}[1]{Section~\ref{sec:#1}}

\newcommand{\lemref}[1]{Lemma~\ref{lem:#1}}
\newcommand{\thmref}[1]{Theorem~\ref{thm:#1}}

\newcommand{\defref}[1]{Definition~\ref{def:#1}}

\newcommand{\nc}{\newcommand}
\nc\tleft{\overleftarrow{t}}
\nc\tright{\overrightarrow{t}}
\nc\ra{\rightarrow}
\nc\la{\leftarrow}

\nc{\half}{\tfrac{1}{2}}
\def\bpm#1\epm{\begin{pmatrix}#1\end{pmatrix}}

\def\bbF{\mathbb{F}}

\def\bbZ{\mathbb{Z}}

\def\begsub#1#2\endsub{\begin{subequations}\label{eq:#1}\begin{align}#2\end{align}\end{subequations}}
\nc\qand{\qquad\text{and}\qquad}
\nc\mnb[1]{\medskip\noindent{\bf #1}}

\newcommand{\bqc}[1]{\raisebox{-3pt}{\begin{qcircuit} #1\end{qcircuit}}}

\begin{document}

\title{Sparse Quantum Codes from Quantum Circuits}

\author{Dave Bacon}
\affiliation{Department of Computer Science and Engineering, University of Washington, Seattle, WA 98195 USA}
\altaffiliation[Current affiliation: ]{Google Inc., Mountain View, CA, USA}
\author{Steven T.\ Flammia}
\affiliation{Centre for Engineered Quantum Systems, School of Physics, The University of Sydney, Sydney, NSW, Australia}
\author{Aram W.\ Harrow}
\affiliation{Center for Theoretical Physics, Massachusetts Institute of Technology, Cambridge, MA, USA }
\author{Jonathan Shi}
\affiliation{Department of Computer Science, Cornell University, Ithaca, NY, USA}
\thanks{This work was presented in part at STOC 2015.}

\date{\today}

\begin{abstract}
  We describe a general method for turning quantum circuits into sparse quantum subsystem
  codes. The idea is to turn each circuit element into a set of low-weight gauge
  generators that enforce the input-output relations of that circuit element.  Using this
  prescription, we can map an arbitrary stabilizer code into a new subsystem code with the
  same distance and number of encoded qubits but where all the generators have constant
  weight, at the cost of adding some ancilla qubits. With an additional overhead of
  ancilla qubits, the new code can also be made spatially local.

  Applying our construction to certain concatenated stabilizer codes
  yields families of subsystem codes with constant-weight generators
  and with minimum distance $d = n^{1-\eps}$, where $\eps =
  O(1/\sqrt{\log n})$. For spatially local codes in $D$ dimensions we
  nearly saturate a bound due to Bravyi and Terhal and achieve $d =
  n^{1-\eps-1/D}$.
  Previously the best code distance achievable with
  constant-weight generators in any dimension, due to Freedman, Meyer
  and Luo, was $O(\sqrt{n\log n})$ for a stabilizer code.
\end{abstract}

\maketitle


\section{Introduction}

Sparse quantum error-correcting codes obey the simple constraint that only a constant 
number of qubits need to be measured at a time to extract syndrome bits. Considerable effort 
has been devoted to studying sparse quantum codes, most notably in the context of 
topological quantum error correction~\cite{Terhal2014}. This effort is driven by the fact that the 
sparsity constraint is quite natural physically, and existing fault-tolerant 
thresholds~\cite{Raussendorf2007} and overheads~\cite{Gottesman13} are optimized when 
the underlying code is sparse. Despite this effort, finding families of \emph{good} sparse 
quantum codes -- i.e.\ codes with asymptotically constant rate and relative distance -- remains 
an open problem, in stark contrast to the situation for classical codes 
(see e.g.~\cite{MacKay2003}). 

Quantum subsystem codes~\cite{Poulin2005} form a broad generalization of standard 
stabilizer codes where a subset of logical qubits is sacrificed to allow for extra gauge 
degrees of freedom. The two principle advantages of subsystem codes are that the 
measurements needed to extract syndrome information are in general sparser and the errors 
only need to be corrected modulo gauge freedom, which often improves fault-tolerance 
thresholds~\cite{AC07}. 

In this paper, we consider a general recipe that constructs a sparse quantum subsystem code 
for every Clifford quantum circuit. The new code resembles the circuit in that the layout of the 
circuit is replaced with new qubits in place of the inputs and outputs of each of the circuit 
elements. The gauge generators are localized to the region around the erstwhile circuit 
elements, and thus the sparsity $s$ of the new code is constant when the circuit is composed 
of few-qubit gates. 

When the circuit satisfies a relaxed form of fault-tolerance and implements a syndrome measurement circuit for a 
``base'' quantum stabilizer code encoding $k$ qubits with distance $d$, then the new sparse 
subsystem code \emph{inherits} the $k$ and $d$ parameters of the base code.
We construct circuits of the requisite special form, and from this we show how 
\emph{every} stabilizer code can be mapped into a sparse subsystem code with the same $k$ 
and $d$ as the original base code. 

The number of physical qubits $n$ required for the new code is roughly the circuit size, and 
this can be chosen to be proportional to the sum of the weights of the original stabilizer 
generators, when permitting circuit elements to be spatially non-local.
This result is summarized in our main theorem.

\begin{theorem}\label{thm:main}
Given any $[n_0,k_0,d_0]$ quantum stabilizer code with stabilizer
generators of weight $w_1,\ldots,w_{n_0 - k_0}$, 
there is an associated $[n,k,d]$ quantum subsystem code whose gauge generators have 
weight $O(1)$ and where $k=k_0$, $d=d_0$, and $n = O\bigl(\sum_i w_i\bigr)$. This mapping is 
constructive given the stabilizer generators of the base code.
\end{theorem}

The proof is in \secref{ft-gadgets}. There are two intermediate
results: \thmref{dist-from-FT} shows the validity of our 
circuit-to-code construction and \thmref{ftgadget} constructs an
efficient measurement gadget that satisfies the relaxed fault-tolerance condition,
using expander graphs.  While expander graphs
have played an important role in classical error correction, to our
knowledge this is their first use in quantum error correction.

We then demonstrate the power of \thmref{main} by applying it to two natural scenarios: 
first to concatenated codes and then to spatially local codes. By applying our construction to 
concatenated stabilizer codes, we obtain families of sparse subsystem codes with by far the 
best distance to date. The previous best distance for a sparse quantum code was due to 
Freedman, Meyer, and Luo~\cite{FML02}, who constructed a family of stabilizer codes 
encoding a single logical qubit having minimum distance $d=O(\sqrt{n\log n})$. Our 
construction provides for the following improvement in parameters.

\begin{theorem}\label{thm:catcodes}
Quantum error correcting subsystem codes exist with gauge generators of weight $O(1)$ and 
minimum distance $d = n^{1-\epsilon}$ where $\epsilon = O(1/\sqrt{\log{n}})$.
\end{theorem}

It is natural to ask if our construction can also be made spatially local. By spatially local we mean that all of the qubits can be arranged on the vertices of a square lattice in $D$ dimensions with each gauge generator having support in a region of size $O(1)$. Incorporating spatial locality is indeed also possible, though it will in general increase the size of the circuit we use, and hence the total number of qubits in the subsystem code. 

\begin{theorem}\label{thm:localcodes}
Spatially local subsystem codes exist in $D\ge 2$ dimensions with gauge generators of weight 
$O(1)$ and minimum distance $d = n^{1-\eps-1/D}$ where $\eps = O(1/\sqrt{\log{n}})$.
\end{theorem}

Although the spatial locality constraint comes at the cost of decreased performance in the rate 
and relative distance, this scaling of the distance is nearly optimal. Several upper bounds have 
been proven about the parameters of spatially local subsystem codes in $D$ dimensions. For 
this case, Bravyi and Terhal~\cite{Bravyi2009} have shown that $d \le O(n^{1-1/D})$. Our 
codes nearly saturate this bound and have the virtue that they are in general constructive. In 
particular, our codes in $D=3$ dimensions already improve on the previous best results (by 
Ref.~\cite{FML02} again) for \emph{arbitrary} sparse codes and achieve $d = n^{2/3-\eps}$ 
for $\eps = O(1/\sqrt{\log{n}})$.

Furthermore, for the class of local commuting projector codes in $D$ dimensions 
(a class that generalizes stabilizer codes, but does \emph{not} contain general subsystem 
codes), Bravyi, Poulin, and Terhal~\cite{Bravyi2010a} have shown the inequality 
$k d^{2/(D-1)} \le O(n)$. It is open whether a similar upper bound holds for subsystem codes,
but a corollary of our main results is that there are spatially local subsystem codes 
for every $D\ge 2$ that obey $k d^{2/(D-1)} \ge \Omega(n)$. 

The remainder of the paper is organized as follows. In Section~\ref{S:background} we review 
the theory of subsystem codes and the prior art. We define the construction for our 
codes in Section~\ref{S:construction} and review the relevant properties of the construction in 
Section~\ref{sec:properties}. Those sections provide a proof of Theorem~\ref{thm:main} 
conditional on the existence of certain fault-tolerant circuits for measuring stabilizer-code 
syndromes, which we subsequently show exist in Section~\ref{sec:ft-gadgets}, thus 
completing the proof. Sections~\ref{S:nonlocal} and \ref{S:local} are devoted to the proofs of 
Theorems~\ref{thm:catcodes} and \ref{thm:localcodes} respectively, and we conclude with a 
discussion of open problems in Section~\ref{S:conclusion}.


\section{Background and Related Work}\label{S:background}
\subsection{Quantum Subsystem Codes}

For a system of $n$ qubits, we can consider the group $\cP^n$ of all
$n$-fold tensor products of single-qubit real-valued Pauli operators $\{I, X, iY,
Z\}$ and including the phases $\{\pm1\}$.   
A \emph{stabilizer code} (see e.g.~\cite{Nielsen2000})
is the joint $+1$ eigenspace of a group of commuting Pauli operators
$\cS = \langle S_1,\ldots,S_l\rangle$, where the $S_i$ label a
generating set for the group. (To avoid trivial codes, we require that
$-I \not\in \cS$.)  
If each of the $l$ generators are independent, then
the code space is $2^{k}$-dimensional where $k = n-l$, and there exist
$k$ pairs of \emph{logical operators} which generate a group $\cL =
\langle X_1, Z_1, \ldots, X_k, Z_k \rangle$.  In general, the logical
group is isomorphic to $C(\cS)/\cS$ where $C(\cS)$ is the centralizer
of $\cS$ in $\cP^n$, meaning the set of all Pauli operators that commute with each element of $\cS$.
The logical group is isomorphic to $\cP^k$, meaning that for each logical
operator in $\cL$ we have that $[L_i, L_j ] = 0$ for all $i\not=j$,
and $X_i Z_i = - Z_i X_i$ for all $i$.  The fact that $\cL \subseteq
C(\cS)$ means that 
$[L_i,S_j] =0$ for all $S_j \in \cS$.
The \emph{weight} of a Pauli operator is the number
of non-identity tensor factors, and the \emph{distance} of a code is
the weight of the minimum weight element among all possible
non-trivial logical operators (i.e.\ those which are not pure
stabilizers).

A \emph{subsystem code}~\cite{Poulin2005, Kribs2006} is a generalization of a stabilizer 
code where we ignore some of the logical qubits and treat them as ``gauge'' degrees of freedom. 
More precisely, in a subsystem code the stabilized subspace $\cH_S$ further 
decomposes into a tensor product $\cH_S = \cH_L \ot \cH_G$, where by convention we still 
require that $\cH_L$ is a $2^k$-dimensional space, and the space $\cH_G$ contains the 
unused logical qubits called gauge qubits. The gauge qubits give rise to a \emph{gauge group} 
$\cG$ generated by former logical operators $G_i$ (which obey the Pauli algebra 
commutation relations for a set of qubits) together with the stabilizer operators. We note that 
$-I$ is always in the gauge group, assuming that there is at least one
gauge qubit. The logical operators in a subsystem code are given by 
 $\cL = C(\cG)$ and still preserve the code space.  The center of the
 gauge group $Z(\cG)$ is defined to be the subgroup of all elements in
 $\cG$ that commute with everything in $\cG$.  Since $Z(\cG)$ contains
 $-I$, it cannot be the set of stabilizers for any nontrivial
 subspace.  Instead we define the stabilizer subgroup
 $\cS$ to be isomorphic to $Z(\cG)/\{\pm I\}$.  Concretely, if
 $Z(\cG)$ has generators $\langle -I, S_1,\ldots,S_l\rangle$ then we
 define the stabilizer group to be $\langle \eps_1 S_1,\ldots, \eps_l
 S_l\rangle$ for some arbitrary choice of $\eps_1,\ldots,\eps_l \in
 \{\pm 1\}$.  
While we focus on codes, define more generally a ``stabilizer subspace'' to be the
simultaneous $+1$-eigenspace of some set of commuting Pauli operators.  

A classic example of a subsystem code is the Bacon-Shor code~\cite{Bacon2006} 
having physical qubits on the vertices of an $L\times L$ lattice (so $n = L^2$). 
The gauge group is generated by neighboring pairs of $XX$ and $ZZ$ operators across the 
horizontal and vertical links respectively. The logical quantum information is encoded by a 
string of $X$ operators along a horizontal line and a string of $Z$ operators along a vertical 
line, and the code distance is $L = \sqrt{n}$. 

We differentiate between two types of logical operators in a subsystem code: \emph{bare}
logical operators are those that act trivially on the gauge qubits,
while \emph{dressed} logical operators may in general act nontrivially
on both the logical and gauge qubits.  In other words, the bare
logical group is $C(\cG) / \cS$ while the dressed logical group is $C(\cS)/\cS$.
The distance of a subsystem code is the minimum weight among
all nontrivial dressed logical operators, i.e. $\min\{|g| : g\in
C(\cS)-\cS\}$.  We say that a code is a $[n,k,d]$ code if it uses $n$  
physical qubits to encode $k$ logical qubits and has distance $d$.

\subsection{Sparse Quantum Codes and Related Work}\label{S:related}

The sparsity of a code is defined with respect to a given set of gauge generators. If each 
generator has weight at most $s_g$ and each qubit partakes in at most $s_q$ generators, 
then we define $s = \max\{s_g,s_q\}$ and say the code is $s$-sparse. We call a code family 
simply \emph{sparse} if $s=O(1)$. 

The most important examples of sparse codes are topological stabilizer
codes, also called homology codes because of their relation to
homology theory. The archetype 
for this code family is Kitaev's toric code~\cite{Kitaev2003}, which encodes $k=O(1)$ qubits 
and has minimum distance $d=O(\sqrt{n})$ (although it can correct a constant fraction of 
random errors). It is known that 2$D$ homological codes obey 
$d \le O(\sqrt{n})$~\cite{Fetaya2012}. Many other important examples of such codes are 
known; see Ref.~\cite{Terhal2014} for a survey. 

The discovery of subsystem codes~\cite{Poulin2005, Kribs2006} led to the study of sparse 
subsystem codes, first in the context of topological subsystem codes, of which there are now 
many examples~\cite{Bacon2006, Bombin2010, Crosswhite2010, Brell2011, Suchara2011, Sarvepalli2012, Bravyi2012a, Bombin2014, Brell2014}. 
These codes are all concerned with the case $k=O(1)$ and all achieve distance
$d=O(\sqrt{n})$.   

Work on codes with large $k$ initially focused on random codes, where it was shown that 
random stabilizers have $k, d \propto n$~\cite{Calderbank1996, Calderbank1998, ABKL00}, 
and more recently that short random circuits generate good codes (meaning that $k$ and $d$
are both proportional to $n$)~\cite{Brown2013}. 
There are also known constructive examples of good stabilizer codes such as those 
constructed by Ashikhmin, Litsyn, and Tsfasman~\cite{Ashikhmin2001} and 
others~\cite{Chen2001, Chen2001a, Matsumoto2002, Li2009}. All of these codes have 
stabilizer generators with weight $\propto n$, however.

A growing body of work has made simultaneous improvement on increasing
$k$ and $d$ while keeping the code sparse. The best distance known to
be achievable with a sparse code is due to Freedman, Meyer and
Luo~\cite{FML02}, encoding a single qubit with distance
$O(\sqrt{n\log n})$.  The only other candidate for beating distance $\sqrt{n}$ with a
sparse code is Haah's cubic
code~\cite{Haah11} which achieves distance $n^\alpha$ where
currently the only known bounds for $\alpha$ are that
$0.01 \leq \alpha \leq 0.56$~\cite{BravyiH11,BravyiH13,Haah-email}.  A
different construction called hypergraph product codes by Tillich and
Z\'{e}mor~\cite{TZ09} achieves a distance of $O(\sqrt{n})$ but with
constant rate. These codes, like the toric code, can still correct a
constant fraction of random errors~\cite{KP13} but they abandon
spatial locality in general. 

Some notion of spatial locality can be recovered by working with more
exotic geometries than a simple cubic lattice in Euclidean
space. Z\'{e}mor constructed a family of hyperbolic surface codes with
constant rate and logarithmic distance~\cite{Zemor2009}; see
also~\cite{Kim2007a}. Guth and Lubotzky~\cite{GL13} have improved this
by constructing sparse codes with constant rate and
$d=O(n^{3/10})$. 
Hastings~\cite{Hastings14} has shown that the 4D toric code in hyperbolic space 
can be made to have a constant rate with a distance of only $d=O(\log(n))$, 
but with an efficient decoder.
These codes and those of Ref.~\cite{FML02} live most
naturally on cellulations of Riemannian manifolds with non-Euclidean
metrics and unfortunately cannot be embedded into a cubic lattice in
$D\le 3$ without high distortion.

The Bacon-Shor codes~\cite{Bacon2006} mentioned in the previous section were 
generalized by Bravyi~\cite{Bravyi2011b} to yield a family of sparse subsystem codes 
encoding $k \propto \sqrt{n}$ qubits while still respecting the geometric locality of the gauge 
generators in $D=2$ dimensions and maintaining the distance $d= \sqrt{n}$. This is an 
example of how subsystem codes can outperform stabilizer codes under spatial locality 
constraints, since two-dimensional stabilizer codes were proven in~\cite{Bravyi2011b} to 
satisfy $kd^2 \leq O(n)$ (which generalizes~\cite{Bravyi2010a} in $D$ dimensions to 
$k d^{2/(D-1)} \le O(n)$). Bravyi~\cite{Bravyi2011b} has also shown that all 
spatially local subsystem codes in $D=2$ dimensions obey the bound $k d \le O(n)$ for $D=2$ 
and so this scaling is optimal for two dimensions.

A family of $O(\sqrt{n})$-sparse codes called homological product codes, due to Bravyi and 
Hastings~\cite{BH13}, leverage random codes with added structure to create good stabilizer 
codes with a nontrivial amount of sparsity, but no spatial locality of the generators. 

By way of comparison, classical sparse codes exist that are able to achieve linear rate and distance, and can be encoded and decoded from a constant fraction of errors in linear time~\cite{Spielman96}. 


\section{Constructing the Codes}\label{S:construction}

In this section we describe how our new codes are constructed.  Our
codes are built from existing stabilizer codes, and indeed our
construction can be thought of as a recipe for sparsifying
stabilizer codes. Our requirements for the initial base code are
described in \secref{base}. Then our construction is presented in
\secref{construction}. 

\subsection{The Base Code and Error-Detecting Circuits}\label{sec:base}

Our code begins with an initial code called $C_0$ which is a
stabilizer code with stabilizer group $\cS_0$ and logical group $\cL_0$. By a slight abuse 
of notation, we use $C_0$ to also refer to the actual code space. It uses $n_0$ qubits to 
encode $k_0$ logical qubits with distance $d_0$. Assume that there exists an 
error-detecting circuit consisting of the following elements:
\bit
	\item A total of $n_a$ ancilla qubits initialized in the $\ket 0$ state.
	\item A total of $n_0$ data qubits.
	\item A Clifford unitary $U_{\ED}$ applied to the data qubits and ancillas.
	\item Single-qubit postselections onto the $\ket 0$ state.
\eit
Denote the resulting operator $V_{\ED}$. By ordering the qubits appropriately we have
\be V_{\ED} = 
(I^{\ot n_0} \ot \bra{0}^{\ot n_a}) U_{\ED}(I^{\ot n_0} \ot \ket{0}^{\ot n_a})\,.
\ee
This satisfies $V^\dag_{\ED} V^\vd_{\ED} \leq I$ automatically.  
\begin{definition}\label{def:error-detect}
A circuit $V_{\ED}$ is a \emph{good error-detecting circuit} for $C_0$ if $V_{\ED}^{\dag} V_{\ED}^{\vd}$ is proportional (with nonzero constant of proportionality) to the projector onto $C_0$.
\end{definition}
This means that the circuit always rejects states orthogonal to $C_0$ and accepts every state in $C_0$ with the same nonzero probability, assuming no errors occur while running the circuit. 
In other words, 
\be
V_{\ED}^\dag V^\vd_{\ED} = 
(I^{\ot n_0} \ot \bra{0}^{\ot n_a}) U_{\ED}^\dag (I^{\ot n_0} \ot \proj{0}^{\ot n_a})
U^\vd_{\ED} (I^{\ot n_0} \ot \ket{0}^{\ot n_a})
\propto  \frac{1}{|\cS_0|}\sum_{s \in \cS_0} s \,.
\label{eq:VV-proj}\ee

We allow the initializations and postselections to occur at staggered times across the circuit, 
so that the circuit is \emph{not} simply a rectangular block in general.
Describing this in sufficient detail for our purposes will require somewhat cumbersome notation.
All initializations, postselections, and elementary gates take place at \emph{half-integer time steps}.
Thus, a single-qubit gate acting  at time, say, $t = 2.5$, can be thought of as mapping the state from
time $t = 2$ to $t=3$.
The $i$th qubit is input or initialized at time $T^\init_i-0.5$ and output or measured at time $T^\out_i+0.5$. 
The total depth of the circuit is then $\max_i T^\out_i - \min_i T^\init_i + 1$.

We defer a discussion of fault-tolerance in our circuits until Sec.~\ref{sec:ft-gadgets}.

\subsection{Localized codes}\label{sec:construction}

To construct our code, we place a physical qubit at each
integer spacetime location in the circuit. Thus, each wire of
the circuit now supports up to $T$ physical qubits, and in general the
$i$th wire will hold $T^\out_i - T^\init_i +1$ physical qubits.
Assume each $T^\init_i \geq 0$ and let $T = \max_i T^\out_i$.  Then
each qubit is active for some subset of times $\{0,\ldots,T\}$.   In
some of our analysis it will be convenient to pretend that each qubit
is present for the entire time $\{0,\ldots,T\}$, and that all
initializations and measurements happen at times -0.5 and $T+0.5$
respectively.  During the ``dummy'' time steps the qubits are acted
upon with identity gates.  It is straightforward to see that the code
properties (except for total number of physical qubits) are identical
with or without these dummy time steps.  Thus, we will present our
proofs as though dummy qubits are present, but will perform our
resource accounting without them.

\begin{table}[t!]
\renewcommand\arraystretch{2.8}
\renewcommand\tabcolsep{10pt}
\begin{center}
\begin{tabular}{cc}
\hline\hline
\textbf{Circuit element} & \textbf{Gauge generators} \\ 
\hline 
    \bqc{\row{0} \qout \gate{I} \qin;} & \raisebox{3pt}{$XX, ZZ$} \\
    \bqc{\row{0} \qout \gate{H} \qin;} & \raisebox{3pt}{$ZX, XZ$} \\
    \bqc{\row{0} \qout \gate{\sqrt{Z}} \qin;} & \raisebox{3pt}{$YX, ZZ$} \\
    \raisebox{-12pt}{\bqc{\row{0} \qout \fcb\d \qin; \row{1} \qout \plus \qin;}} & 
	${XX \atop \!X\, I}, {\,I \,\,I \atop XX}, {ZZ \atop I \,I }, {Z\, I \atop ZZ}$ \\
    \bqc{\row{0} \qbra \qw;} & \raisebox{3pt}{$Z$}\\
    \hspace{8pt}\bqc{\row{0} \qket;} & \raisebox{3pt}{$Z$} \\
\hline\hline
\end{tabular}
\end{center}
\renewcommand\arraystretch{1}
\caption{Dictionary for transforming circuit elements into generators of the gauge group. For every input and output of a circuit element in the left column, we add the corresponding generators from the right column, placed on the appropriate physical qubits. (This is the purpose of the $\eta^i_t$ map in the main text.) We only list the gauge generators for the standard generators of the Clifford group, but the circuit identities of any Clifford circuit can be used instead. Pre- and postselections are special and have only one gauge generator associated to them.}
\label{T:dictionary}
\end{table}

We introduce the function $\eta^i_t(P)$ to denote placing a Pauli $P$ at
spacetime position $(i,t)$.  If $P$ is a multi-qubit Pauli then we
let $\eta^i(P)$ or $\eta_t(P)$ denote placing it either on row $i$
coming from circuit qubit $i$ or on column $t$ corresponding to
circuit time slice $t$.  For a two-qubit gate $U$, we write
$\eta^{i,j}_t(U)$ to mean that we place $U$ at locations $(i,t)$ and
$(j,t)$.  When describing a block of qubits without this spacetime
structure, we also use the more traditional notation of $P_i$ to denote
Pauli $P$ acting on position $i$: that is, $P_i := I^{\ot i-1} \ot P
\ot I^{\ot n-i}$, where $n$ is usually understood from context.

With this notation in hand, we define the gauge group of our codes. 
This is summarized in Table~\ref{T:dictionary}, and defined more precisely below. 
The gauge group will have $2k$ generators per $k$-qubit gate and one
for each measurement or initialization. 
Let $U$ be a single qubit gate that acts on qubit $i$ as it
transitions from time $t$ to time $t+1$.  Corresponding to this gate,
we add the gauge generators
\begin{subequations}\label{eq:gauge-gen}
\be
\eta^i_{t+1}(U^\vd X^\vd U^\dag) \eta^i_t(X^\vd) \qand
\eta^i_{t+1}(U^\vd Z^\vd U^\dag) \eta^i_t(Z^\vd).\ee
Similarly for a two-qubit gate $U$ acting on qubits $i,j$ at time
$t+1/2$, we add the generators
\be\{ \eta^{i,j}_{t+1}(UP U^\dag)\eta^{i,j}_t(P) : P \in\{X \ot I, Z\ot I, I\ot X,
I\ot Z\}_{i,j}\}\,. \ee
More generally, a $k$-qubit gate $U$ acting on qubits $i_1,\ldots,i_k$
at time $t+1/2$ has generators
\be\{ \eta^{i_1,\ldots,i_k}_{t+1}(UP U^\dag)\eta^{i_1,\ldots,i_k}_t(P)
: P = I^{\ot j-1} \ot Q \ot I^{\ot k-j}, j \in [k], Q \in \{X,Z\} \}\,. \ee
\end{subequations}
For measurements or
initializations of qubit $i$ we add generators $\eta^i_{T^\out_i}(Z)$ or $\eta^i_{T^\init_i}(Z)$
respectively. 

An illustration of the mapping from the circuit to the code is given in Fig.~\ref{F:circuit2code}.

\begin{figure}[t!b]
\begin{center}
\includegraphics[width=.7\textwidth]{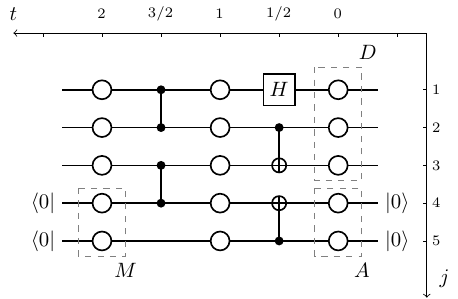}
\caption{Illustration of the circuit-to-code mapping. Using integral spacetime coordinates 
$(j,t)$, the open circles at integer time steps $(j,t)$ are physical qubits of the subsystem code, 
while gates of the circuit are ``syncopated'' and live at half-integer time steps $(j,t)$. 
The three sets of qubits in the 
dashed boxes labelled $D$, $A$, and $M$ correspond to the input qubits for the base 
code, the ancillas, and the measurements (postselections)
respectively.   For this circuit, for example, we have among 
others the gauge generators  $\eta_1^1(X) \eta_1^0(Z)$ and
$\eta_1^1(Z) \eta_1^0(X)$
because these are the circuit identities for the Hadamard gate at spacetime location $(1,1/2)$. Note 
that we pad each line with identity gates to ensure that there are always an even number of 
gates on each line, which is important to maintain our code properties.
We draw our circuit diagram with time moving from right to left to
match the way that operators are composed; e.g. if we apply $U_1$ then
$U_2$ then $U_3$ the resulting composite operator is $U_3 U_2 U_1$.
(see Section~\ref{sec:properties} for details).
\label{F:circuit2code}}
\end{center}
\end{figure}


\section{Code Properties}\label{sec:properties}

In this section we prove that our codes match---in the sense of
Theorem~\ref{thm:main}---the performance of the base codes with
respect to $k$ and $d$.  There are four intermediate steps.  In
\secref{circuit-code} we explain the way in which our localized codes
inherit the properties of the circuits used to build them.  Then in
\secref{iso-groups} we will show that they inherit the properties of
the original code used.  Specifically we will show that the logical
and stabilizer groups for the new code are in a sense isomorphic to
the corresponding groups of the base code.  Next, \secref{distance}
will show that the new code inherits the distance properties of the
original code {\em if} the error-detecting circuit satisfies a relaxed fault-tolerance condition.
\secref{ft-gadgets} describes a
general construction of fault-tolerant measurement gadgets (loosely based on
the DiVincenzo-Shor syndrome measurement scheme~\cite{Shor96,DS96})
satisfying the conditions we need.
\subsection{The circuit-code isomorphism} \label{sec:circuit-code}

We describe first a change of basis for \eq{gauge-gen} that will be
useful later. Let $U_{t+1/2}$ 
denote the global unitary that is applied between time step $t$ and
time step $t+1$, i.e.\ the product of all gates occurring at time $t+1/2$.
Then we can rewrite \eq{gauge-gen} as
\be \{ \eta_{t+1}(U_{t+1/2} P U_{t+1/2}^\dag)\eta_t(P) : P \in
\cP^{n_0 + n_a}\}.\label{eq:gauge-gen-1}\ee
This is a massively overcomplete set of generators, some of which have 
high weight. But we will use it only in the analysis of the code, and
will not measure the generators in \eq{gauge-gen-1} directly.

However, this form is useful to show that every operator that acts on the
physical qubits is gauge-equivalent to one that acts just on the qubits at 
the first time-step of the circuit. Because the action of the gauge group 
cleans up the circuit by sliding all errors in one direction, we call this the squeegee lemma.
\begin{lemma}[The Squeegee Lemma]\label{lem:squeegee}
There is a map $\varphi: \cP^{\ot (T+1)(n_0+n_a)} \to \cP^{\ot (T+1)(n_0+n_a)}$ such that $\varphi(P)$ is in the image of $\eta_0$ and also $P\,\varphi(P)\in\cG$ for every Pauli operator $P$.

Furthermore $P\,\varphi(P)$ is a product of gauge generators corresponding to unitary gates.
\end{lemma}
\begin{proof}
Say that an operator $Q$ is localized to times $\leq t$ if $Q$ is in the image of $\eta_{0,\ldots,t}$.
Then $P$ is localized to times $\leq T$.
Let $Q|_t$ denote the restriction of an operator $Q$ on the circuit to time slice $t$.
Then define
\be \label{eq:squeegee-induct}
\varphi_t(Q) := Q\, \eta_t(Q|_t)\,\eta_{t-1}(U_{t-1/2}^\dag Q|_{t}U_{t-1/2}).
\ee 
Since $\eta_t(Q|_t)$ cancels the slice-$t$ component of $Q$, if $Q$ was localized to times $\leq t$ then $\varphi(Q)$ is localized to times $\leq t-1$.
And by \eq{gauge-gen-1}, we see that
$Q\,\varphi(Q) \in\cG$ and is in fact a product of gauge generators corresponding to unitary gates in the circuit.
So take $\varphi = \varphi_1 \circ \dots \circ  \varphi_T$.
Then $\varphi(P)$ is localized to time slice $0$ and $P\,\varphi(P)$ is a product of unitary gate gauge generators.
\end{proof}

Next, the set $\cP^{n_0+n_a}$ is invariant under conjugation
by Cliffords. Thus we can replace \eq{gauge-gen-1} with
\be 
  \{ \eta_{t+1}(U_{\leq t+1/2} P U_{\leq t+1/2}^\dag)
\eta_t(U_{\leq t-1/2} P U_{\leq t-1/2}^\dag) : P \in
\cP^{n_0 + n_a}\} \,,
\label{eq:gauge-gen-2}
\ee
where $U_{\leq t+1/2} := U_{t+1/2} U_{t-1/2} \cdots U_{1/2}$.
We could equivalently take $P$ to range over only the individual $X$
and $Z$ operators.  This time, though, the action of $U_{\leq t+1/2}$
would make the resulting operators still high weight.    Thus the
generators in \eq{gauge-gen-2} should also not be measured directly, but will be useful for 
proving certain operators are contained in $\cG$.

Overall the gauge group can be interpreted as representing the action of the error-detection 
circuit $V_{\ED}$.  Each generator of the gauge group in \eq{gauge-gen} is meant to 
propagate a single Pauli through a single gate of the circuit. The generators in 
\eq{gauge-gen-1} and \eq{gauge-gen-2}  correspond to propagating a column of Paulis 
through one time step of the circuit. More generally, we could consider the effects of the 
circuit on an arbitrary pattern of Paulis occurring throughout the circuit.

\begin{definition}\label{def:pattern}
Define a Pauli pattern (also sometimes called an error pattern) $E$ to be a collection of Pauli 
matrices, one for each circuit location. Let $E_{i,t}$ denote the Pauli occurring at time $t$ 
(assumed to be integer) and qubit $i$. We can interpret $E$ either as a sequence of Paulis 
that occur at different times, or as the unitary $E = \prod_{i,t} \eta_t^i (E_{i,t})$. Let 
$E_t := \bigotimes_i E_{i,t}$. For a unitary circuit $U = U_{T-1/2} \cdots U_{1/2}$ define 
$U_E := E_T U_{T-1/2} E_{T-1} \cdots U_{1/2} E_0$.
\end{definition}

Armed with this definition, we can see one form of the isomorphism between a Clifford 
circuit $V$ and the corresponding gauge group $\cG$ with the following lemma. 
\begin{lemma}\label{lem:circuit-gauge-equiv}
Let $\cG$ be the gauge group corresponding to a Clifford circuit $V$. 
For any error pattern $E$ and any gauge pattern $G$ (i.e.\ the Pauli pattern corresponding 
to an element of $\cG$) we have
\be 
  V_E  = \pm V_{GE} \,.
\label{eq:circuit-gauge-equiv}
\ee
\end{lemma}

\begin{proof}
By induction it suffices to consider the case when $G$ is a single
gauge generator of the form \eq{gauge-gen}.
Write $V = A_f U_{T-1/2} \cdots U_{1/2} A_i$ where $A_f = I^{\ot n_0}
\ot \bra{0}^{n_a}$ and $A_i = A_f^\dag$.  
Thus we have
$$V_{GE} = A_f (G_T E_T) U_{T-1/2} (G_{T-1}E_{T-1}) \cdots U_{1/2} (G_0 E_0)
A_i.$$
Suppose that $G$ corresponds to a single unitary gate at time $t+1/2$
so that the only nonidentity $G_\tau$ are for $\tau = t,
t+1$.  Then the expansions of $V_E$ and $V_{GE}$ differ only
surrounding the gate $U_{t+1/2}$ where they look like
\be 
  E_{t+1} U_{t+1/2} E_t \qand 
 G_{t+1} E_{t+1} U_{t+1/2} G_{t}E_{t}\,,
\ee
respectively.  But $G_{t+1}$ and $G_t$ are not independent since they are in 
the gauge group. Namely, 
$G_{t+1} = U_{t+1/2} G_t U_{t+1/2}^\dag$, 
as can be seen from the defining equations \eq{gauge-gen}. 
Then depending on whether $E_{t+1}$ and $G_{t+1}$ 
commute ($+$) or anticommute ($-$), we have
\be
  G_{t+1} E_{t+1} U_{t+1/2} G_{t}E_{t} = 
  \pm E_{t+1} G_{t+1} U_{t+1/2} G_{t}E_{t} = 
  \pm E_{t+1} U_{t+1/2} E_{t} \,.
\ee
Thus $V_E$ and $V_{GE}$ differ by at most a sign.

If $G$ is the generator corresponding
to an initialization or measurement, then the situation is similar. 
For a measurement, $G= \eta_T^i(Z)$ and $A_f G = A_f$. For an
initialization, $G= \eta_0^i(Z)$ and we need to commute $G$ past $E_0$
before it reaches the $A_i$ and is annihilated.  That is $G_0 E_0 A_i
= \pm E_0 G_0 A_i = \pm E_0 A_i$.

Since this covers a complete generating set, we conclude that
\eq{circuit-gauge-equiv} holds for all $G\in \cG$.
\end{proof}

\subsection{Characterizing the logical and stabilizer groups}\label{sec:iso-groups}

\subsubsection{Spackling}

We can think of our codes as being based on the following primitive:
for each section of wire between gates in a given quantum circuit, we
simply replace that wire with a physical qubit. More precisely, we
interpret each circuit element as consisting of a gate together with
input and output wires, and allow adding explicit identity gates. We
note that initializations only have output wires, while postselections
only have input wires. For adjacent circuit elements, an output wire
of the preceding element is identical to the input wire of the
following element, to avoid double-counting.

The isomorphism between the base code and the new code is based 
on the operation of taking some Pauli operator on the circuit input, 
propagating it through to each time step of the circuit, and leaving 
a copy of the operator at each time step in the appropriate 
spacetime locations in the new code. We call this operation
``spackling'' because it resembles the process of spreading a line of
putty over a 2-dimensional region and letting it dry in a site-specific way.

Here is a more concrete definition of spackling. Suppose that
$U_{\ED} = U_{T-1/2} U_{T-3/2} \ldots U_{1/2}$, where $U_t$ is the product of all
gates performed in time step $t$. Equivalently, $U_{\ED} = U_{\le T-1/2}$. 
Let $P$ be a Pauli operator on the input qubits (including ancillas)
of the circuit, i.e.\ $P \in \cP^{n_0 + n_a}$.  Then 
\be 
  \spack(P) := \prod_{t=0}^T \eta_t(U_{\leq t-1/2} P U_{\leq t-1/2}^\dag) \,.
  \label{eq:spack-product}
\ee
Here we adopt the convention that $U_{-1/2}=I$.

A simple example of spackling comes from a plain wire with no gates, or more precisely, 
with all identity gates. Then $\spack(P) = P^{\ot T+1}$ for $P  \in \cP$.  Observe that the
$\spack$ map preserves commutation relations between Pauli operators if $T$ is even.
More generally \eq{spack-product} shows how $\spack$ preserves the commutation 
relations of the Pauli algebra on the input qubits. In this sense it can be viewed as 
an injective homomorphism.  

Another useful property of the spackling map is that its outputs are gauge-equivalent to its 
inputs. Again this relies on $T$ being even, and can be illustrated on a single wire, 
where $\eta_0(P)$ is gauge-equivalent to $P^{\ot T+1}$ for each $P\in\cP$. 
More generally we have:
\begin{lemma}\label{lem:spack-equiv}
  For $P\in\cP^{n_0  + n_a}$, $\spack(P) \eta_0(P)\in \cG$.
\end{lemma}
\begin{proof}
By first \eq{spack-product} and then the fact that $T$ is even, we have
\ba \spack(P) \cdot  \eta_0(P)
 & =\prod_{t=1}^T \eta_t(U_{\leq t-1/2} P U_{\leq t-1/2}^\dag)\\
& = \prod_{t'=1}^{T/2}
\eta_{2t'}(U_{\leq 2t'-1/2} P U_{\leq 2t'-1/2}^\dag)
\eta_{2t'-1}(U_{\leq 2t'-3/2} P U_{\leq 2t'-3/2}^\dag).\ea
By \eq{gauge-gen-2}, this is contained in the gauge group.
\end{proof}

\subsubsection{Proving the isomorphism}

Recall that our base code consists of a code $C_0$ with $n_0$ physical
qubits, $k_0$ logical qubits, and distance $d_0$, together with an
error-detecting circuit for fault-tolerantly measuring the stabilizer
generators. From this base code, we follow a sequence of intermediate
codes and build our final code in stages. For $i=1,2,3,4$ we define
code $C_i$ with gauge group $\cG_i$, stabilizer group $\cS_i \cong
Z(\cG_i)/\{\pm I\}$ and logical group $\cL_i = C(\cG_i)/\cS_i$, and
standard code parameters $n_i$, $k_i$, and $d_i$. We then build our
code from the circuit for $C_0$ in the following stages: 
\bit
  \item $C_1$ consists of just wires, i.e., a circuit with only identity gates. 
  \item $C_2$ adds the unitary gates from the circuit.
  \item $C_3$ adds the initializations.
  \item $C := C_4$ adds the postselections.
\eit

Let's analyze the code properties of each of these. There are $n_1 =
n_0 + n_a$ lines (wires), each consisting of an odd numbers of qubits,
corresponding to circuits in which each line has an even number of
circuit elements. In $C_1$, suppose the wires have lengths $l_1,
\ldots, l_{n_0 + n_a}$.  The gauge group $\cG_1$ is generated by
$\{XX,ZZ\}$ on nearest-neighbor qubits along wires. The logical group
is $\cL_1 = \langle \{\eta^i(X^{\ot l_i}), \eta^i(Z^{\ot l_i}) : i\in
[n_0 + n_a]\}\rangle$.  The stabilizer group $\cS_1$ is empty. These
claims rely on the fact that each $l_i$ is odd, since if the $l_i$
were even then the elements of the logical group would commute with
each other.

To get $C_2$, we add the actual unitary gates from the circuit. The
new gauge group $\cG_2$ is related to $\cG_1$ by conjugation. We do
not conjugate each element of $\cG_1$ by the same unitary. Instead,
label the generators $\eta_{t+1}^i(X)\eta_t^i(X),
\eta_{t+1}^i(Z)\eta_t^i(Z)$ of $\cG_1$ by a half-integer time $t$. For each $P
\in \{X,Z\}^{\ot n_1}$ and each half-integer gate time $t$ there is a generator
$\eta_{t+1}(P)\ot \eta_{t}(P)$ of $\cG_1$.  For $\cG_2$ we replace this generator
with $\eta_{t+1}(UPU^\dag)\ot \eta_{t}(P)$, as illustrated in Table \ref{T:dictionary}. 
For example, if we apply a Hadamard gate, then the generators $X\ot X,Z\ot Z$ 
are replaced with $Z\ot X, X\ot Z$.  For a CNOT gate, we replace the generators 
\be 
  \begin{array}{cc}X & X \\ I & I\end{array}\qquad
  \begin{array}{cc}Z & Z \\ I & I\end{array}\qquad
  \begin{array}{cc}I & I \\ X & X\end{array}\qquad
  \begin{array}{cc}I & I \\ Z & Z\end{array}
\ee
with the generators
\be 
  \begin{array}{cc}X & X \\ X & I\end{array}\qquad
  \begin{array}{cc}Z & Z \\ I & I\end{array}\qquad
  \begin{array}{cc}I & I \\ X & X\end{array}\qquad
  \begin{array}{cc}Z & I \\ Z & Z\end{array}\,.
\ee
The logical group is similarly twisted and now is 
\be
  \cL_2 = \langle \{\spack(P_i) : P \in \{X,Z\}, i \in [n_0+n_a]\} \rangle\,,
\ee
which follows from the next lemma.
\begin{lemma}\label{lem:spackle-logical}
For any $P$ acting on the input qubits, $\spack(P)$ commutes with all elements of $\cG_2$.
Furthermore, an operator on the circuit commutes with all elements of $\cG_2$ \emph{only if} it is equal to $\spack(P)$ for some $P$.
\end{lemma}
\begin{proof}
Let $g \in \cG_2$ be a gauge generator, corresponding to a gate at some time $t+1/2$.
Let $Q|_t$ denote the restriction of some Pauli operator $Q$ to the $t^{\text{th}}$ time step.
By definition, $\spack(P)|_t = \left(\prod_{\tau=0}^t U_{\tau-1/2}\right)^{\vd}P\left(\prod_{\tau=0}^t U_{\tau-1/2}\right)^{\dag}$ and similarly for $\spack(P)|_{t+1}$.
It follows that $\spack(P)|_{t+1} = U_{t+1/2}^{\vd}\spack(P)|_tU_{t+1/2}^{\dag}$.

By definition, also $g = \eta_{t+1}(U_{t+1/2}^{\vd}QU_{t+1/2}^{\dag})\eta_t(Q)$ for some Pauli $Q$.
Since $g$ is only supported on time steps $t$ and $t+1$, it suffices to verify that $g$ commutes with $\spack(P)|_t\spack(P)|_{t+1}$.
This is verified by noting that conjugation by $U_{t+1/2}$ is a group isomorphism and hence preserves commutation/anticommutation relations.

To argue the other direction, suppose $Q$ is some operator which is not a spackle.
We note that the condition $\spack(P)|_{t+1} = U_{t+1/2}^{\vd}\spack(P)|_tU_{t+1/2}^{\dag}$ for every $P$ and $t$ is actually equivalent to the definition of $\spack$.
Thus there is some $t$ such that $Q|_{t+1} \ne U_{t+1/2}^{\vd}Q|_tU_{t+1/2}^{\dag}$.
Since $Q|_{t}( U_{t+1/2}^{\dag}Q|_{t+1}U_{t+1/2}^{\vd})$, is localized to the $t$-th time step and is not equal to the identity, there is some Pauli operator $A$ so that $\eta_t(A)$ anticommutes with it.
It follows that $\eta_t(A)$ anticommutes with exactly one of $Q|_{t}$ or $U_{t+1/2}^{\dag}Q|_{t+1}U_{t+1/2}^{\vd}$.
Therefore $Q$ anticommutes with $g = \eta_{t+1}(U_{t+1/2}^{\vd}A U_{t+1/2}^{\dag})\eta_t(A)$.
\end{proof}

The stabilizer group $\cS_2$ is still empty.  The following lemma is
an immediate consequence of the squeegee lemma (\lemref{squeegee}) 
and \lemref{spack-equiv}.
\begin{lemma}\label{lem:full-spackle}
Let $P$ be a Pauli operator acting on the input qubits. Then $\eta_0(P)$, $\spack(P)$, and 
$\eta_T(U_{\ED} P U_{\ED}^\dag)$ are all equivalent modulo $\cG_2$. 
That is, the product of any pair of them is contained in $\cG_2$.
\end{lemma}

Adding initializations brings us to $C_3$. Of course, there is nothing fundamentally different 
from our code's perspective between initializations and postselections, since both consist of 
adding $Z$ operators locally, and because our gates are reversible. The only asymmetry 
comes from our decision to analyze initializations first. 
Each initialization removes one generator from the 
logical group and adds it to the stabilizer group. (The conjugate element of the logical group 
becomes the corresponding error.) The initializations commute with each other.  Now the 
logical and stabilizer groups are
\ba
  \cL_3 &= \langle \spack(P_i) : P \in \{X,Z\}, i \in [n_0] \rangle
  \label{eq:L3}\\
  \cS_3 &= \langle \spack(Z_{n_0+i}) : i \in [n_a] \rangle \,.
  \label{eq:S3}
\ea
Every new element of the gauge group is equivalent (modulo $\cG_2$) 
to one of these stabilizers, so the new gauge group is simply 
$\cG_3 = \langle \cG_2 \cup \cS_3\rangle$. 

So far we have not used any specific properties of the circuit. The 
final code $C = C_4$ adds the postselections, aka measurements. Let 
$M$ be the set of qubits where postselections take place (in our case always equal to the
set of ancilla qubits), 
and let $\cM$ be the group generated by $\{Z_i : i\in M\}$. Then
$\cG_4 = \langle \cG_3 \cup \eta_T(\cM) \rangle$.  (Here we have used the assumption that
measurements all occur at the final time $T$.)  In general adding elements 
to the gauge group will cause the elements to move from the logical 
group to the stabilizer group and from the stabilizer group to the 
(pure) gauge group (i.e.\ $\cG/\cS$).

Before proving the next lemma, we introduce the following notation. For any operator $A$ of appropriate dimension, we define a map $u$ by
\begin{align}
	u(A) := U_{\ED}^\dag A U^{\vd}_{\ED} \,.
\end{align}
Now we have the following lemma.

\begin{lemma}
	\label{lem:good-circuit-decomp}
	Let $r_0$ and $r_a$ denote the restriction
	maps taking $\cP^{\ot (n_0+n_a)}$ to the first $n_0$ or the last $n_a$ qubits respectively.
	If
	\[V_{\ED} = 
	(I^{\ot n_0} \ot \bra{0}^{\ot n_a}) U_{\ED}(I^{\ot n_0} \ot \ket{0}^{\ot n_a}) \]
	is a good error-detecting Clifford circuit for $C_0$ and if $\cM$ is the group generated by the $Z$ operators on ancilla qubits (i.e.~$Z_j$ for $n_0+1 \le j \le n_0+n_a$), then there exists a pair of subgroups $\cH$ and $\cH^{\perp}$ of $\cM$ such that
\benum
\item
    $\cH \cdot \cH^{\perp} = \cM$ and $\cH \cap \cH^\perp = \{I\}$;
\item
	$r_0(u(\cH)) = \cS_0$.
\item
	$r_a(u(\cH)) \subseteq \cM$.
\item
	Every non-identity element of $u(\cH^{\perp})$ anticommutes with some element of $\cM$.
\eenum
\end{lemma}
\begin{proof} Since $\cM$ is
  isomorphic (as a group) to $\bbZ_2^{n_a}$, for every subgroup $\cA$ of $\cM$ there is a
  subgroup $\cA^\perp$ such that $\cA \cdot \cA^\perp = \cM$ and $\cA \cap \cA^\perp =
  \{I\}$.

Define $\cH = \cM\cap u^{-1}(C(\cM))$.  In
  other words, $\cH$ is the subgroup of $\cM$ consisting of $H \in \cM$ such that $r_a(u(H))$
  centralizes $\cM$.  Then every non-identity
  element of $u(\cH^\perp)$ anticommutes with some element of $\cM$.
  Already
  we can see that parts 1, 3 and 4 of the Lemma are true.
In the rest of the proof we will show that $\cH$ corresponds to stabilizers and
$\cH^\perp$ to gauge qubits, on the way to establishing part 2 of the Lemma.
	
  For $A \in \cM$, let $\Phi_{A}$ denote the projector to the $+1$ eigenspace of $A$, and
  for $\cA$ a subgroup of $\cM$, let $\Phi_{\cA}$ denote the projector to the common $+1$
  eigenspace of every element in $\cA$.  Since $\cM$ is abelian we have 
\be 
\Phi_{\cA}
  =\prod_{A\in \cA}\Phi_A = \frac{1}{|\cA|} \sum_{A\in\cA} A.\ee
  More generally if $\cA,\cB$ are subgroups of $\cM$ then
  $\Phi_{\cA\cB} = \Phi_{\cA}\Phi_{\cB}$.

  For any operator $\Theta$, define 
\be \pi(\Theta) = 
(I^{\ot n_0} \ot \bra{0}^{\ot n_a}) \Theta (I^{\ot n_0} \ot
\ket{0}^{\ot n_a}).\ee
For $a,b \in \{0,1\}^{n_0}, a',b'\in \{0,1\}^{n_a}$ we can calculate
\be \pi(X^aZ^b \ot X^{a'}Z^{b'}) = X^aZ^b \mathbbm{1}_{a'=0},\ee
where $\mathbbm{1}_p$ is the indicator function for proposition $p$.
In particular, for a Pauli $P$, $\pi(u(P))=r_0(u(P)) \mathbbm{1}_{P\in\cH}$.
It follows that for a subgroup $\cA\subseteq \cM$ we have
\be \pi(u(\Phi_\cA)) = \frac{1}{|\cA|}\sum_{A\in\cA}  \pi(u(A))
= \frac{|\cA\cap \cH|}{|\cA|} r_0(u(\Phi_{\cA\cap \cH})).
\label{eq:proj-subgp}\ee

This implies that
\begsub{VdagV}
V_{\ED}^\dag V_{\ED} ={}& \pi(u(\Phi_\cM)) = \frac{|\cH|}{|\cM|} r_0(u(\Phi_\cH)) \\
&\pi(u(\Phi_\cH)) =  r_0(u(\Phi_\cH)) \label{eq:phi_ch}
\\ 
&\pi(u(\Phi_{\cH^\perp}))
= \frac{1}{|\cH^\perp|} I^{\ot n_0}
= \frac{|\cH|}{|\cM|} I^{\ot n_0}
\endsub
Since $V_{\ED}^\dag V_{\ED}$ is proportional to the projector onto
$C_0$ and $r_0(u(\Phi_{\cH}))$ is a projector, this implies $r_0(u(\Phi_{\cH}))$ is equal to the projector onto the codespace $C_0$.

We can now establish item 2 of the Lemma.
This follows from $r_0(u(\Phi_{\cH})) = \frac{1}{|\cH|}\sum_{H \in \cH} r_0(u(H))$ being equal to the projector to the codespace
of $C_0$, which is also equal to $\frac{1}{|\cS_0|}\sum_{s \in \cS_0} s$ for $\cS_0$ the stabilizer group of $C_0$.
\end{proof}

We will also need a lemma about group centralizers. 
Recall that $C(A)$ denotes the centralizer of a subgroup $A$ (with respect to the larger group $\cP^{\ot n})$.
\begin{lemma}
	\label{lem:centralizer}
	For any two subgroups $A$ and $B$ of a larger group, 
	\[C(A B) = C(A) \cap C(B)\,.\]
	Also, if the larger group is $\cP^{\ot n} = \langle I, X, iY, Z\rangle^{\ot n}$, and both $A$ and $B$ contain $-I$, then additionally, 
	\[C(A \cap B) = C(A)C(B)\,.\]
\end{lemma}
\begin{proof}
	To show $C(A B) = C(A) \cap C(B)$, note that since $A$ and $B$ are subsets of $AB$, an element centralizes $AB$ only if it centralizes both $A$ and $B$.
	Also, if an element centralizes both $A$ and $B$, then it must also commute through any element of $AB$.
	
	We will establish the intermediate result that $C(C(A)) = A$ if $A$ is itself a centralizer of some subgroup.
	Suppose that $x \in A$.
	Then by the definition of the centralizer, $x$ commutes with every element in $C(A)$ and thus $x \in C(C(A))$.
	Note that this shows $A' \subseteq C(C(A'))$ for all sets $A'$ even if $A'$ is not a centralizer.
	Suppose now that $A=C(A')$ and $x \not\in A$.
	Then there must be some element $y \in A'$ such that $x$ does not commute with $y$.
	But $y \in C(C(A')) = C(A)$, since we've established $A' \subseteq C(C(A'))$.
	Since $x$ does not commute with $y$ and $y \in C(A)$, we must have $x \not\in C(C(A))$.
	So $C(C(A)) \subset A$ if $A$ is a centralizer. 
	
	Now if $A$ and $B$ are both centralizers, then
	\[C(A \cap B) = C(\,C(C(A)) \cap C(C(B))\,) =  C(\,C(C(A)C(B))\,) = C(A)C(B)\,.\]
	
	And every subgroup $A$ of $\cP^{\ot n}$ that contains $-I$ is a centralizer.
	To see this, note that $A$ is a logical subgroup of some stabilizer code, and by making every non-logical operator a gauge operator, we see that $A$ is the centralizer of that gauge group.
\end{proof}

We now characterize the stabilizer group $\cS = \cS_4$ and logical group $\cL= \cL_4$ of the final code $C_4$, after adding the postselections.

\begin{lemma}[Isomorphism of codes] \label{lem:isomorphism}
Let $\cM$, $r_0$, and $r_a$ be defined as in \lemref{good-circuit-decomp}, so that $r_0$ and $r_a$ are the restriction
maps taking $\cP^{\ot (n_0+n_a)}$ to the first $n_0$ or the last $n_a$ qubits respectively,
and $\cM = \langle Z_j : n_0+1 \leq j \leq n_0+n_a\rangle \subseteq \cP^{\ot (n_0 + n_a)}$.
Further let $\cM_a = r_a(\cM)$ denote the restriction of $\cM$ to the ancilla qubits.
Then: \begin{enumerate}
	\item
The stabilizer group $\cS_4 = \spack(\cK_{\cS})$ for some group $\cK_{\cS} \subseteq \cP^{\ot(n_0+n_a)}$ such that $r_0(\cK_{\cS}) = \cS_0$ and $r_a(\cK_{\cS}) \subseteq \cM_a$.
Furthermore, $\cK_{\cS} = \left\langle C(\cM) \cap u(\cM),\; \cM \cap C(u(\cM))\right\rangle$.
	\item
The bare logical group $\cL_4 = \spack(\cK_{\cL})$ for some group $\cK_{\cL} \subseteq \cP^{\ot(n_0+n_a)}$ such that $r_0(\cK_{\cL}) = \cL_0$ and $r_a(\cK_{\cL}) \subseteq \cM_a$.
Furthermore $r_0$ is a bijection between $\cK_{\cL}$ and $\cL_0$.
	\item
As a consequence of the above points, an error pattern $E$ is a logical operator on $C_4$ if and only if $(V_{\ED})_E = V_{\ED}E_{\mathrm{in}}$ for some Pauli $E_{\mathrm{in}}\in\cL_0$.
\end{enumerate}
\end{lemma}
In other words, the stabilizers and logical operators of $C_4$ are spackled versions of the stabilizers and logical operators of original code, plus there are some additional stabilizers, corresponding to spackling of some subgroup of $\cM_a$.

\begin{proof}
We'll start by characterizing $C(\cG_4)$, since this equals $\cS_4\cdot \cL_4$.

By \lemref{spackle-logical}, the centralizer of $\cG_2$ was precisely the set of operators $\spack(\cP^{n_0+n_a})$.
When adding the initialization operators in the circuit, we had $\cG_3 = \langle\cG_2 \cup
\eta_0(\cM)\rangle$ and the centralizer of $\cG_3$ was $\spack(C(\cM)) = \spack(\cP^{\ot n_0} \ot \cM_a) $.
Finally, when adding the postselections, $\cG_4 = \langle \cG_3 \cup \eta_T(\cM)
\rangle$.
Equivalently by \lemref{full-spackle}, $\cG_4 = \langle \cG_3 \cup \eta_0(u(\cM))\rangle$.
Therefore,
\newcommand{\phantomlemmaeq}{\stackrel{\hphantom{\text{\lemref{centralizer}}}}{=}}
\begsub{centralizer}
C(\cG_4) &\stackrel{\text{\lemref{centralizer}}}{=} C(\cG_3) \cap C\left(\eta_0(u(\cM))\right)
\\&\phantomlemmaeq \spack\left(C(\cM)\right) \,\cap\, C\left(\eta_0(u(\cM))\right)
\\&\phantomlemmaeq \spack\bigl(C(\cM) \cap C(u(\cM))\bigr)\,.
\label{eq:centralizer-stabilizer}
\endsub
This last equality used the fact that $\spack(P)$ commutes with $\eta_0(Q)$ if and only if $P$ commutes with $Q$.
 
We therefore make the definition 
\begin{align}
	\cK := C(\cM) \cap C(u(\cM))\,.
\end{align}
Note that $\cK \subseteq C(\cM) = \cP^{\ot n_0} \ot \cM_a$ so that $r_a(\cK) \subseteq \cM_a$.
Next we take the subgroups $\cH$ and $\cH^{\perp}$ as in \lemref{good-circuit-decomp},
so that $\cM = \cH\cdot \cH^{\perp}$, $r_0(u(\cH)) = \cS_0$, 
and $r_a(u(\cH)) \subseteq \cM_a$. 
This implies that $\cM\,u(\cH) = \cS_0 \ot \cM_a$. 
Now from the definition of $\cK$ and using \lemref{centralizer} we find 
\begsub{centralizer-2}
\cK &\phantomlemmaeq C(\cM) \cap C(u(\cH \cdot \cH^{\perp}))
\\&\stackrel{\text{\lemref{centralizer}}}{=} C(\cM) 
\cap C(u(\cH)) \cap C(u(\cH^{\perp}))
\\&\stackrel{\text{\lemref{centralizer}}}{=} C(\cM\,u(\cH)) 
\cap C(u(\cH^{\perp}))
\\&\phantomlemmaeq 
C(\cS_0 \ot \cM_a) \cap C(u(\cH^{\perp}))
\,.
\endsub
By Item 4 of \lemref{good-circuit-decomp}, $C(\cM) \cap u(\cH^{\perp}) = \{I^{\ot n_0 + n_a}\}$.
By \lemref{centralizer}, since $C(C(\pm \cM)) = \pm \cM$,
\begin{align}
C(u(\cH^{\perp}))\cM
= C\left(C(\cM) \cap \pm u(\cH^{\perp})\right)
= C(\{I^{\ot n_0 + n_a}\}) = \cP^{\ot n_0 + n_a}.
\end{align}
We now are ready to characterize $\cK$, at least modulo $\cM$.  
\be
\cK\cM =
C(\cS_0 \ot \cM_a)\cM \cap C(u(\cH^{\perp}))\cM = 
C(\cS_0 \ot \cM_a)\cM = 
C(\cS_0) \ot \cM_a\,.
\ee
Thus $r_0(\cK) = C(\cS_0) = \cS_0 \cdot \cL_0$ and $r_a(\cK) \subseteq \cM_a$.

We have established that $\cS_4 \cdot \cL_4 = C(\cG_4) = \spack(\cK)$,
and $r_0(\cK) = \cS_0 \cdot \cL_0$ with $r_a(\cK) \subseteq \cM_a$.
While our goal is to characterize both the stabilizer and logical subgroups, we argue that it suffices to prove Item 1 of the Lemma about the stabilizer group.  In particular, 
we would like to show that there is a normal subgroup $\cK_{\cS}$ of $\cK$
such that $\spack(\cK_{\cS}) = \cS_4$ and $r_0(\cK_{\cS}) = \cS_0$.  Since
$\spack$ is an injective homomorphism, this would suffice to determine $\cK_{\cL} \cong \cK/\cK_{\cS}$ with $\spack(\cK_{\cL}) = \cL_4$ and $r_0(\cK_{\cL}) = \cL_0$ as well, up to multiplication by elements in $\cS_4$ and $\cS_0$ respectively. 
Therefore, let us define
\begin{align}
	\cK_{\cS} := \left\langle C(\cM) \cap u(\cM),\; \cM \cap C(u(\cM))\right\rangle \,.
\end{align}
This is a subgroup of $\cK$ because $\cM$ is an abelian group. 
We will show that $\spack(\cK_{\cS}) = \cS_4$, and also that $\cK_{\cS}$ is precisely the subgroup of $\cK$
whose elements $K$ satisfy $r_0(K) \in \cS_0$.

First we will show that $\cK_{\cS}$ is the largest subgroup of $\cK$ satisfying $\eta_0(\cK_{\cS}) \subseteq \cG_4$.
Let us define $\Delta = \langle \cM\cup u(\cM)\rangle$, and embed this group into the full circuit via $\eta_0$. 
Then we have that $\eta_0(\Delta) = \cG_4 \cap \im \eta_0$.
The fact that both $\cM$ and $u(\cM)$ are abelian groups then 
implies that $\cK_{\cS} = \cK \cap \Delta$. 
Furthermore, by repeated applications of \lemref{full-spackle} we find that 
\be \cG_4 = \langle \spack(\Delta) \cup \cG_2  \rangle \label{eq:G4-Delta}.\ee
Next we claim that 
\be
\label{eq:imspack}
\cG_2 \cap  \im \spack = \{I^{\otimes(n_0 + n_a) (T+1)}\}\,.
\ee
This is true because if $\spack(P) \in \cG_2$ then by \lemref{full-spackle} we have $\eta_0(P) \in \cG_2$ and $\cG_2 \cap \im \eta_0 = \{I^{\otimes (n_0+n_a)(T+1)}\}$, implying $P = I^{\ot n_0 + n_a}$.
Putting this together we obtain
\begsub{S4-char}
\cS_4 & \stackrel{\phantom{\text{\eq{G4-Delta}}}}{=} C(\cG_4) \cap \cG_4 \\
& \stackrel{\text{\eq{centralizer}}}{=} \spack(\cK) \cap \cG_4 \\
& \stackrel{\text{\eq{G4-Delta}}}{=} \langle \spack(\cK \cap \Delta) \cup (\spack(\cK) \cap \cG_2)\rangle \\
& \stackrel{\text{\eq{imspack}}}{=} \spack(\cK \cap \Delta)  \\
& \stackrel{\phantom{\text{\eq{imspack}}}}{=} \spack(\cK_{\cS})\,.
\endsub
This establishes the first claim in Item 1 of the lemma. 

The fact that $\cK_{\cS}$ is precisely the subgroup of $\cK$
whose elements $K$ satisfy $r_0(K) \in \cS_0$ comes from the definition of $\cK_{\cS}$ combined with Items 2 and 4 of \lemref{good-circuit-decomp}, which say that the largest part of $u(\cM)$ that is contained in $C(\cM)$ has an image under $r_0$ which is contained in $\cS_0$. 
This completes the remaining claims in Item 1. 

To show uniqueness for the logical operators, if $K$ and $K'$ were distinct elements of $\cK_{\cL}$ with $r_0(K) = r_0(K')$, then $\spack(KK')$ would also be a bare nonidentity logical operator. 
However $r_0(KK') = I$, so $KK'\in\cM$ and thus $\spack(KK') \in \cG_3 \subseteq \cG_4$.  This contradicts our assumption that $K\neq K'$, and thus establishes Item 2 of the lemma.

To obtain the last sentence of the lemma statement, we combine the above characterization of the bare logical operators
and stabilizers with \lemref{full-spackle}, which tells us that all logical operators $E$ in $C_4$ are gauge-equivalent to $\eta_0(E_{\text{in}}
\ot Q)$ for some $Q \in \cM_a$ and some $E_{\text{in}} \in \cL_0$.
Then \lemref{circuit-gauge-equiv} shows
that the error pattern $E$ is gauge-equivalent to $\eta_0(E_{\text{in}} \ot Q)$ if and only if $(V_{\ED})_E
= V_{\ED}E_{\text{in}}$.
\end{proof}

In what follows we write $\cG := \cG_4$, $\cS := \cS_4$, $\cL := \cL_4$, since we will 
focus on the final code construction and not the intermediate steps used in the proof. 

\subsection{Distance}\label{sec:distance}

\begin{definition}
Given a collection of errors $E = (E_t)_t$, define the weight $|E_t|$ to be 
the number of nonidentity terms in $E_t$ and define $|E| = \sum_t |E_t|$.  
\end{definition}

\begin{definition}\label{def:FT}
A Clifford subcircuit $V$ is a fault-tolerant sub-projection when
\begin{enumerate}
	\item $V=c\Pi$ where $0< c\leq 1$ and $\Pi$ is a projection onto a stabilizer subspace.
	\item
	For any error pattern $E$, either $V_E = 0$ or there are Pauli errors $E_{\mathrm{in}}$ and $E_{\mathrm{out}}$ such that $V_E = E_{\mathrm{out}}VE_{\mathrm{in}}$ and $|E_{\mathrm{in}}E_{\mathrm{out}}| \le |E|$.
\end{enumerate}
\end{definition}

The second part of the definition is a relaxed form of the condition from quantum fault-tolerance that
errors should not propagate to become equivalent to larger errors (explained in \cite{gottesman2009introduction}).
One interpretation of this condition is that every error pattern $E$ factors into the
product of an incorrect measurement outcome for the code syndrome (the occurrence of any
such error would be equivalent to $PVP$ for some Pauli $P$) with some error of weight $\le |E|$ on the circuit's input or output.
Another interpretation is that no subcircuit of a fault-tolerant circuit transforms the logical subspace of the base code into a smaller-distance subspace.
This is sufficient for the error-correcting code, since it ensures that the logical subspace is preserved at high distance by the gauge transformations associated with the circuit gates.
And in fact any circuit error that results in an incorrect syndrome measurement---but otherwise leaves the circuit's data untouched and does not cause any postselection rejections---is itself in the gauge group of the code, as it is a circuit identity.

We note also some easy consequences of being proportional to a stabilizer projector.
\begin{lemma}\label{lem:stab-proj}
Let $V$ be a Clifford sub-projection, i.e.~satisfying Part 1 of \defref{FT}.  Then:
\benum
\item
 For a Pauli operator $P$, if $V \ne PVP$ then $VPV = 0$.
\item For any two Pauli operators $P$ and $Q$, if $V = QVP$ then also $V = QVQ$ and $V = PVP$.
\eenum
\end{lemma}

\begin{proof}
\benum
\item
The space $\im V$ is a stabilizer subspace, say corresponding to stabilizer subgroup $\langle s_1,
\ldots, s_k\rangle$.  The space $\im
PVP$ is also a stabilizer subspace, this time stabilized by
\be \langle Ps_1P, \ldots, Ps_kP\rangle = \langle (-1)^{a_1}s_1, \ldots,
(-1)^{a_k}s_k\rangle,\ee
where $a_i$ is defined by $Ps_iPs_i = (-1)^{a_i}I$.  Note that $V=PVP$ precisely when $a_i=0$ for all $i$.
Thus if $V\neq PVP$ then at least one $a_i=1$; without loss of generality say
$a_1=1$. Since $\im V\subseteq 
\im (I+s_1)$ and $\im PVP\subseteq \im (I-s_1)$, these subspaces are orthogonal, implying
$0=V(PVP) = VPV$.
\item Since $V$ is a sub-projection we have $cV = V^\dag V = PVQQVP = PV^2P = cPVP$. The
  same argument implies $V=QVQ$. 
\eenum
\end{proof}

A key feature of \defref{FT} is {\em composability}: multiple fault-tolerant circuits will
form a larger fault-tolerant circuit in their composition, as long as all of them commute
with each other. 
\begin{lemma}\label{lem:composability}
If $V^{(1)},\ldots,V^{(T)}$ are mutually commuting fault-tolerant Clifford subcircuits
partitioning a circuit  $U = V^{(T)} \cdots V^{(1)}$ with $U \ne 0$,
then $U$ is fault-tolerant as well.
\end{lemma}

\begin{proof}
Part 1 of \defref{FT} follows directly from induction, so we prove that $U$ satisfies Item 2.
We assume that $T = 2$, noting that all higher values of $T$ follow from induction.
We need the mutual commutation assumption here to ensure that the compositions remain Clifford sub-projections (satisfying Part 1 of \defref{FT}) throughout the induction.

An error $E$ on $U$ partitions into errors $E_1$ and $E_2$ on 
the separate subcircuits.
By part 2 of \defref{FT},
\be U_E = V^{(2)}_{E_2} V^{(1)}_{E_1} 
= \left[E^{(2)}_{\text{out}}V^{(2)}E^{(2)}_{\text{in}}\right]
\left[E^{(1)}_{\text{out}}V^{(1)}E^{(1)}_{\text{in}}\right], 
\ee
where 
\be |E^{(1)}_{\text{in}}E^{(1)}_{\text{out}}| +
|E^{(2)}_{\text{in}}E^{(2)}_{\text{out}}| \le |E_1| + |E_2| =  |E|.
\label{eq:split-wt}\ee
Suppose that $U_E \ne 0$, implying that
$V^{(2)}E_{\text{in}}^{(2)}E_{\text{out}}^{(1)}V^{(1)}\neq 0$.

Since $V^{(1)}$ and $V^{(2)}$ are Clifford sub-projectors, express
$V^{(1)} = c_1\sum_{A \in \cA} A$ and $V^{(2)} = c_2\sum_{B \in \cB} B$
for $\cA$ and $\cB$ abelian subgroups of $\cP^{\ot n}$.
Let $P = E_{\text{in}}^{(2)}E_{\text{out}}^{(1)}$ so that
\[ c_1c_2\sum_{A \in \cA}\sum_{B \in \cB}BPA = \left(c_2\sum_{B \in \cB}B\right)P\left(c_1\sum_{A \in \cA}A\right) \ne 0\,.\]
In particular, this implies that there are no $A \in \cA$ and $B \in \cB$ such that $P = -BPA$.
Since $P$ and $A$ and $B$ are all Pauli operators which either commute or anticommute with each other, this implies $P \in C(\langle A\rangle)$
whenever $A = B^{-1} = B$, and we never have $A = -B$ since $U = V^{(2)}V^{(1)} \ne 0$.
Thus $P \in C(\pm\cA \cap\pm\cB)$, implying by \lemref{centralizer} that $P \in C(\pm\cB)C(\pm\cA) = C(\cB)C(\cA)$.

So $P = RQ$ where $Q \in C(\cA)$ and $R \in C(\cB)$, so that $QV^{(1)} = V^{(1)}Q$ and $RV^{(2)} = V^{(2)}R$.
Therefore:
\begsub{ft-composability}
U_E & = 
E^{(2)}_{\text{out}}V^{(2)}E_{\text{in}}^{(2)}E_{\text{out}}^{(1)}V^{(1)}E^{(1)}_{\text{in}}
\\&= E^{(2)}_{\text{out}}V^{(2)}\,[P]\, V^{(1)}E^{(1)}_{\text{in}}
\\&= E^{(2)}_{\text{out}}V^{(2)}\,[RQ]\, V^{(1)}E^{(1)}_{\text{in}}
\\&= [E^{(2)}_{\text{out}}R]\,V^{(2)}V^{(1)}\,[QE^{(1)}_{\text{in}}]
\\&= [E^{(2)}_{\text{out}}R]\,U\,[QE^{(1)}_{\text{in}}]\,.
\endsub
Finally we bound
$|E^{(2)}_{\text{out}}RQE^{(1)}_{\text{in}}|=|E^{(2)}_{\text{out}}E_{\text{in}}^{(2)}E_{\text{out}}^{(1)}E^{(1)}_{\text{in}}| \leq |E|$
  using \eq{split-wt} and the triangle inequality.
\end{proof}

\begin{theorem}[Main theorem]\label{thm:dist-from-FT}	
If $V$ is a fault-tolerant error-detection circuit (i.e.\ satisfying 
Definitions~\ref{def:error-detect} and \ref{def:FT}) for an $[n,k,d]$ stabilizer code,
then the corresponding localized code constructed in \secref{construction} is an $[O(|V|),k,d]$ subsystem code.

If $w$ is the largest arity of a gate in $V$, then each
gauge generator in the code has weight at most $w+1$ and each qubit is measured by $\leq
2w+2$ gauge generators. 
\end{theorem}

\begin{proof}
The sparsity of the code and the number of qubits used is established in \secref{construction}.
The preservation of the number of encoded qubits $k$ is established by \lemref{isomorphism}.
All that remains is to show that the distance $d$ is also preserved.
	
Let $E$ be a Pauli error on the localized code, which we can think of 
equivalently as an error pattern on the circuit $V$.  Suppose that $E$ 
is a nontrivial dressed logical operator (so that $E \not\in \cG$).
We would like to show that $|E|\geq d$.

By \lemref{isomorphism}, if $E$ is a nontrivial logical operator then $V_E = VE_{\text{in}}$ for some $E_{\text{in}}$ which is a nontrivial logical operator of $C_0$, and in particular $V_E \ne 0$.
By Item 2 of \defref{FT}, there then exist $E_{\text{in}}'$ and $E_{\text{out}}'$ such that $V_E = E_{\text{out}}'VE_{\text{in}}'$ and $|E_{\text{in}}'E_{\text{out}}'| \le |E|$.
By \lemref{stab-proj}, since $VE_{\text{in}} = E_{\text{out}}'VE_{\text{in}}' \ne 0$ we
find that $V = E_{\text{out}}'VE_{\text{in}}'E_{\text{in}} =
E_{\text{out}}'VE_{\text{out}}' \ne 0$. Therefore by left- and right-multiplying both sides of the equality by $E_{\text{out}}'$, we find $V = VE_{\text{in}}'E_{\text{in}}E_{\text{out}}'$.
Since $V$ is a scalar multiple of a projection to the codespace, $E_{\text{in}}'E_{\text{in}}E_{\text{out}}'$ must be a stabilizer of $C_0$.
So since $E_{\text{in}}$ is a nontrivial logical operator of $C_0$ and $E_{\text{in}}'E_{\text{in}}E_{\text{out}}'$ is a stabilizer, $E_{\text{in}}'E_{\text{out}}'$ is a nontrivial logical operator of $C_0$ as well.
Thus $d \le |E_{\text{in}}'E_{\text{out}}'| \le |E|$.
\end{proof}


\section{Fault-tolerant gadgets}\label{sec:ft-gadgets}

The final piece of our construction is a fault-tolerant (by \defref{FT}) gadget for
measuring a single stabilizer generator, say of weight $w$.   With this we can construct
fault-tolerant circuits for any stabilizer code, and therefore 
sparse subsystem codes from any stabilizer code, completing
the proof of \thmref{main}.

The requirements for fault-tolerance here are markedly different from
those surrounding existing fault-tolerant measurement strategies.  
Our circuits are restricted to
stabilizer circuits, and cannot make use of classical feedback or
post-processing. 
On the other hand, the circuits here only need to detect errors rather than correct them, and furthermore only need to detect them up to a global change in Pauli frame.
The gadgets we use are hence
a variation on the DiVincenzo-Shor cat-state
method~\cite{Shor96,DS96}, modified in several ways.

The original cat-state gadgets would prepare a $w$-qubit cat
state $\frac{1}{\sqrt 2}(\ket 0^{\ot w} + \ket 1^{\ot w})$ and perform a
CNOT from each qubit in the cat state to each qubit that we
want to measure.  Making this fault-tolerant requires additional ``ancilla verification''
steps and also requires repeating the measurement multiple times and taking a majority
vote.  This uses non-Clifford gates (to compute majority) and also would require $O(w^2)$
size to measure a $w$-qubit stabilizer generator.

Due to the relaxed fault-tolerance requirement, we may omit the repetition of the measurement.
We may also use postselection gates instead of measurement gates, though in general for Clifford circuits these are equivalent up to a change in Pauli frame depending on measurement outcomes \cite{knill2005quantum}.

One feature of our modified gadget is that it uses $\ket{+}^{\otimes w}$ instead of the cat state.
We call these $w$ qubits ``vertex qubits'' and also introduce ``edge qubits'' which play the
role of ancilla verification.  These vertex and edge qubits correspond to an expander graph with $w$ vertices and
$O(w)$ edges.  If the original stabilizer is $P_1 \cdots P_w$ for single-qubit Paulis
$P_1,\ldots, P_w$ then for each $i$ we will perform a controlled $P_i$ from vertex qubit
$i$ to data qubit $i$.  In the same step we will perform CNOTs from vertex qubit $i$ to
the edge qubits corresponding to incident edges.
We will see that by building this pattern of CNOTs from an expander
graph, the gadget can be made fault-tolerant.

First, we define these terms more precisely.

\begin{definition}\label{def:gadget}
A fault-tolerant postselection gadget for a Pauli $P \in \cP^{\ot w}$ is a fault-tolerant
(Def.\ \ref{def:FT}) subcircuit $V$ with $w$ inputs and $w$ outputs,
such that $V$ is proportional to a projector into the ${+}1$ eigenspace of $P$.
\end{definition}

\begin{definition}
For an undirected graph $G = (V_G,E_G)$, $v\in V_G$ and $S\subseteq V_G$, define $\partial v = \partial \{v\}$ and
define $\partial S$ to be the set of edges having exactly one endpoint in $S$.  Say
that $(V_G,E_G)$ has edge expansion $\phi$ if $|\partial S| \geq \phi |S|$
whenever $|S| \leq |V_G|/2$.
\end{definition}
We can identify subsets of $V_G, E_G$ with the vector spaces
$\bbF_2^{V_G},\bbF_2^{E_G}$, in which case $\partial$ is a $\bbF_2$-linear operator from 
$\bbF_2^{V_G}$ to $\bbF_2^{E_G}$.

\begin{lemma}\label{lem:gadget}
\mbox{} 
\benum
\item 
Given $P \in \cP^{\ot w}$ and a undirected graph $(V_G,E_G)$ 
with 
 degree $d$ and $|V_G|\geq w$, there exists a postselection gadget for $P$
that has size at most $O(|E_G|)$, and is composed of
Clifford gates that act on no more
than $d+2$ wires at a time.  
\item If $(V_G,E_G)$ has edge expansion $\geq 1$ then this gadget is
  fault-tolerant.
\eenum
\end{lemma}

\begin{proof}
%
Our construction uses two blocks of ancilla wires.  There is one block
with $|V_G|$ wires called the {\em vertex block} and another block of $|E_G|$
ancillas that will be the \emph{edge block}.  The original $w$
qubits being measured we refer to as the {\em data block}.  These are
depicted in \fig{ftmeas}.

\begin{figure}[t!]
\includegraphics[width=.75\textwidth]{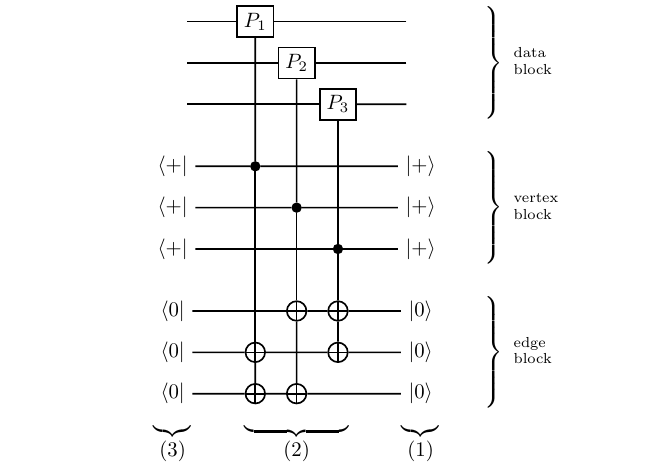}
\caption{
  Example configuration for the fault-tolerant postselection gadget, for
  $w = 3$.
  The data block consists of input wires which are postselected to a ${+}1$
  eigenstate of $P_1 P_2 P_3$. The vertex block is prepared in the
  $\ket{+}^{\otimes 3}$ state, and the edge block is used for parity checks on that 
  state.  Time goes from right to left.
}
\label{fig:ftmeas}
\end{figure}

The circuit consists of the following steps, corresponding to the
labeled steps in \fig{ftmeas}.
\benum
\item The vertex block is initialized to $\ket + ^{\ot |V_G|}$.
  The edge block is initialized to $\ket 0^{\ot |E_G|}$.
\item For each
  $i\in [w]$ perform a controlled-$P_i$ gate from vertex qubit $i$ to data
  qubit $i$.  At the same time, perform a CNOT from vertex qubit $i$ to each edge qubit in
  $\partial v_i$.   Thus each 
of the $w$ vertex qubits is the control for $d+1$ controlled operations.  Since
these commute, we can perform their product as a single gate (for each vertex
qubit).
\item Postselect the vertex block onto 
$\bra{+}^{\ot |V_G|}$ and
  the edge block onto $\bra 0^{\ot |E_G|}$.
 \eenum

The total number of wires in each block is $w$, $|V_G|$, and $|E_G|$ 
for the data, vertex, and edge blocks respectively, 
and since each wire has constant length, the total size of the gadget is 
$O(w)+O(|V_G|)+O(|E_G|) = O(|E_G|)$ as claimed. 

For the analysis, we first split up $V$ into subcircuits
\[V = A_f V^{(|V_G|)} \cdots V^{(1)} A_i\,, \]
where $A_f$ is the postselection to $\ket{0}^{\ot |E_G|}$ on the edge block, $A_i$ is the initialization to $\ket{0}^{\ot |E_G|}$ on the edge block, and for the $v$th wire of the vertex block, $V^{(v)}$ is the subcircuit consisting of the initialization of that wire to $\ket{+}$, the parallelized controlled gate acting on that wire, and finally its postselection to $\ket{+}$.

For each $v \in V_G$, $V^{(v)}$ is a fault-tolerant circuit implementing a projection to the ${+}1$ eigenspace of $P_v\prod_{e \in \partial v} X_e$, where the $P_v$ here is taken to act on the $v$th wire of the data block.
To show that $V^{(v)}$ implements this projection, note that since the state on the $v$th vertex wire is $(\ket{0} + \ket{1})/\sqrt{2}$, and the controlled gate acts as identity when the $v$th vertex wire is $\ket{0}$, we have
\[ V^{(v)} = \frac{1}{2}\left[I^{\ot w} \ot I^{\ot |E_G|}\right] + \frac{1}{2}\left[P_v\prod_{e \in \partial v} X_e\right]. \]
To show that $V^{(v)}$ is fault-tolerant, consider some error pattern $E$ on $V^{(v)}$.
Decompose $E = E_{\text{data}}\ot E_{\text{vertex}}\ot E_{\text{edge}}$ acting on the respective blocks of the gadget.
Since the only sites of $V^{(v)}$ on the data and edge blocks are the inputs and outputs to $V^{(v)}$, we can without loss of generality assume that $E_{\text{data}}$ and $E_{\text{edge}}$ are both identity operators, since any errors on the inputs and output could only increase $|E|$ by at least as much as they increased $|E_{\text{in}}E_{\text{out}}|$.
Since any $X$ operator acts as identity on the vertex block, we only have to consider $Z$ errors for $E_{\text{vertex}}$.
But a single $Z$ on the vertex block makes the circuit evaluate to $0$, and a pair of $Z$ errors propagates to become equivalent to a single $Z$ on the input or the output of any adjacent edge block.
Therefore $|E_{\text{in}}E_{\text{out}}| \le |E_{\text{vertex}}|$ and so $V^{(v)}$ is fault-tolerant.

By \lemref{composability} then, $V^{(|V_G|)} \cdots V^{(1)}$ is a fault-tolerant circuit that projects to the common ${+}1$ eigenspace of $P_v\prod_{e \in \partial v} X_e$ for each $v \in V_G$.
So, letting $P^{S}$ denote $\prod_{v \in S} P_v$ for $S \subseteq V_G$, and using the fact that
$\prod_{v \in S} \prod_{e \in \partial v} X_e = \prod_{e \in \partial S} X_e$,
\begin{align*}
V
&= \left[I^{\ot w} \ot \bra{0}^{\ot |E_G|}\right] \left[\frac{1}{2^{|V_G|}}\sum_{S \subseteq V_G} P^S\prod_{e \in \partial S} X_e\right] \left[I^{\ot w}\ot\ket{0}^{\ot |E_G|}\right]
\\&= \frac{1}{2^{|V_G|}}\sum_{S \subseteq V_G} P^S\cdot \bra{0}^{\ot |E_G|}\left[\prod_{e \in \partial S} X_e\right]\ket{0}^{\ot |E_G|}
\\&= \frac{1}{2^{|V_G|}}\sum_{S \subseteq V_G, \partial S = \emptyset} P^S
\\&= \frac{1}{2^{|V_G|}} (I^{\ot w} + P)
\,,
\end{align*}
where we used the fact that $G$ is connected to conclude that $\emptyset$ and $V_G$ itself are the only subsets $S$ of $V_G$ such that $\partial S=\emptyset$.
Thus $V$ is proportional to a projection onto the ${+}1$ eigenspace of $P$.

Now we show $V$ is fault-tolerant.
Consider some error $E$ on $V$, and we want to show that if $V_E \ne 0$ then there are $E_{\text{in}}$ and $E_{\text{out}}$ so that $V_E = E_{\text{out}}VE_{\text{in}}$ and $|E_{\text{in}}E_{\text{out}}| \le |E|$.
By the fault-tolerance of $V^{(|V_G|)} \cdots V^{(1)}$ and \lemref{composability}, either $V_E = 0$ or
\be V_E = A_fE_fV^{(|V_G|)} \cdots V^{(1)}E_iA_i\ee
so that $|E_iE_f|\le |E|$.
The only sites where $E_i$ and $E_f$ could act on the data block are the inputs and outputs of $V$, where they could only contribute to $|E_iE_f|$ and therefore $|E|$ by at least as much as they contribute to $|E_{\text{in}}E_{\text{out}}|$, so we can assume that $E_i$ and $E_f$ do not act on the data block.
We can assume that $E_i$ and $E_f$ contain no $Z$ errors on the edge block because those act as identity immediately following an initialization to $\ket{0}$ or preceding a postselection to $\ket{0}$.
Since $X$ errors on the edge block also commute through $V^{(|V_G|)} \cdots V^{(1)}$, we find that without loss of generality,
\be V_E = A_fV^{(|V_G|)} \cdots V^{(1)}E_iA_i \ee
where $|E_i| \le |E|$ and $E_i$ consists only of $X$ operators on the edge block.

Therefore, letting $W$ be the set of edge block wires that $E_i$ acts on,
\begin{align*}
V_E
&= \frac{1}{2^{|V_G|}}\sum_{S \subseteq V_G} P^S\cdot \bra{0}^{\ot |E_G|}\left[\prod_{e \in \partial S} X_e\right]E_i\,\ket{0}^{\ot |E_G|}
\\&= \frac{1}{2^{|V_G|}}\sum_{S \subseteq V_G, \partial S = W} P^S
\,.
\end{align*}
If there is no $S$ so that $\partial S = W$, then $V_E = 0$, so we assume such an $S$ exists.
Since $G$ is connected, this $S$ must be unique up to taking its complement.
Since $G$ has edge expansion $\phi\geq 1$, either $|S| \le |W|$ or $|S^{c}| \le |W|$.
Without loss of generality, choose $S$ so that $|S| \le |S^c|$ and $|S| \le |W|$.
Since $P^{S^c} = P^SP$, we have
\begin{align*}
V_E
&= \frac{1}{2^{|V_G|}} (P^S + P^SP)
\\&= P^S \frac{1}{2^{|V_G|}} (I^{\ot w} + P)
\\&= P^S\, V
\,,
\end{align*}
where $|P^S| \le |S| \le |W| = |E_i| \le |E|$.
We conclude that if $V_E \ne 0$, then $V_E = E_{\text{out}}VE_{\text{in}}$ for some $E_{\text{in}}$ and $E_{\text{out}}$ satisfying $|E_{\text{in}}E_{\text{out}}| \leq |E|$, so that the gadget is fault-tolerant.
\end{proof}

\begin{theorem}\label{thm:ftgadget}
For each $P \in \cP^w$, there exist fault-tolerant postselection gadgets
composed of Clifford gates that act on at most $10$ wires at a time, and which
have size at most $O(w)$.
\end{theorem}
\begin{proof}
For each $w' = 7^{6c} - 7^{2c}$ for some integer $c$, there exists a Ramanujan graph
with degree $d=8$ and $w'$ vertices \cite{Morgenstern199444}.  The
Ramanujan property means that the second largest eigenvalue of the
adjacency matrix is $\leq 2\sqrt{7}$.  By Cheeger's
inequality~\cite{HooryLW06}, this implies that the edge expansion is
$\geq \frac{8 - 2\sqrt{7}}{2} \geq 1.35$.
Choose $c$ to be the minimum that satisfies
$w' \geq w$. Then $w' \le 7^6w \leq O(w)$. 
Therefore, by Lemma \ref{lem:gadget}, there is a fault-tolerant postselection
gadget that acts on at most $d+2 = 10$ wires at a time, and with size at most
$O(w)$. 
\end{proof}

We now complete the proof of \thmref{main}.

\begin{proof}
For each stabilizer generator of weight $w_i$, use \thmref{ftgadget}
to construct a FT postselection gadget for that generator using
$O(w_i)$ qubits.  Concatenating these we obtain a FT error-detecting
circuit for the original stabilizer code, due to
\lemref{composability}.  Thus \thmref{dist-from-FT} implies that the localized code inherits
the distance of the original stabilizer code.
\end{proof}

\begin{remark}[Explicit bounds on constant factors]
	This construction is explicit enough that we can bound the implicit constants 
	in the big-$O$ notation for the gadgets.
	First we remark that while the specific expander graph construction used in the proof of \thmref{main} introduces a large constant factor, more efficient expander families are known to exist with varying levels of constructivity.
	In particular, any randomly chosen regular degree-6 graph has
        sufficient edge expansion with 
        probability $1-o_w(1)$~\cite{Bollobas1988}. 
	
	So, assuming knowledge of a degree-$d$, size-$w'$ edge
        expander, for each gadget measuring an operator of weight
        $w \le w'$, the final code requires (recalling that we must
        pad our wires to odd length) at most $3 w$ qubits for the data block,
        $3 w'$ qubits for the vertex block, and $3 w' d/2$ for the
        edge block, yielding a total of at most $3w'(d/2+2)$
        qubits.
        If we use the degree-8 explicit Ramanujan graphs
        from the proof of \thmref{main} then we have $d=8$ and $w'\leq
        7^6w$ for a total qubit cost of $\leq 2117682w$.  If we use
        degree-6 expanders of with $w$ vertices then the total qubit
        cost is $\leq 15w$, although then the construction is not
        explicit.  Indeed such expander graphs are known only to exist
        for sufficiently large $n$ (no explicit bound is given in
        \cite{Bollobas1988}), so technically we should say that the
        qubit cost is bounded by $\max(C, 16w)$ for some universal
        constant $C$.
	
	Each gauge generator acts on at most $(d+3)$ qubits at a time,
        and each qubit participates in at most
        $\lfloor (d+1)/2\rfloor+4$ gauge generators, as follows: The
        wire in the vertex block touches $d+3$ of the $2(d+2)$ local
        generators for the controlled gate.  Of these, $d+1$ come from
        back-propagation of $Z$ errors, while the other two come from
        an $X$ or $Z$ on the control line.  Half of the $d+1$ $Z$-type
        errors can be chosen to propagate ``backwards'' across the
        gadget instead, leaving one side to see
        $\lceil (d+1)/2\rceil+2$ generators and the other side to see
        $\lfloor (d+1)/2\rfloor+2$.  We can place the padding identity
        gate on the side with fewer generators, increasing the total
        number of gauge generators on that side by two for a total of
        $\lfloor (d+1)/2\rfloor+4$.  On the other side
        we place the initialization or postselection, adding just one
        more generator and also yielding a total of at most
        $\lfloor (d+1)/2\rfloor+4$.
\end{remark}


\section{Sparse quantum codes with improved distance and rate}\label{S:nonlocal}

Our Theorem~\ref{thm:main} implies that substantially better distance can be 
achieved with sparse subsystem codes than has previously been achieved. 
The following argument is based on applying our main theorem to concatenated families of 
stabilizer codes with good distance. It was suggested to us by Sergey Bravyi and we are 
grateful to him for sharing it with us.

To apply this argument, we must first have that codes with good distance exist. This is 
guaranteed by the quantum Gilbert-Varshamov bound, one version of which states that if 
$\sum_{j=0}^{d-1} \binom{n}{j} 3^j \leq 2^{n-k}$ then an $[n,k,d]$ quantum stabilizer code 
exists~\cite{Calderbank1996}. This argument chooses a random stabilizer code, and in 
general the generators will have high weight. 

\begin{theorem}[Restatement of Thm.~\ref{thm:catcodes}]\label{thm:catcodes2}
Quantum error correcting subsystem codes exist with gauge generators of weight $O(1)$ and 
minimum distance $d = n^{1-\epsilon}$ where $\epsilon = O(1/\sqrt{\log{n}})$.
\end{theorem}

\begin{proof}
Begin with a base stabilizer code having parameters $[n_0, 1, d_0]$ where $d_0 = \delta n_0$ 
for some constant relative distance $\delta$. Call this code $C_0$. It has $n_0 -1$ 
independent stabilizer generators each with weight $\leq n_0$.

Let $C_m$ be the concatenation of $C_0$ with itself $m$ times. Then $C_m$ is a 
$[n_0^m, 1, d_0^m]$ code whose stabilizer generators can be classified according to 
their level in the concatenation hierarchy. There are $\le n_0^{m+1-j}$ generators at level $j$, 
each with weight $\le n_0^j$, for a total weight of $\le mn_0^{m+1}$. 

Next we apply Theorem \ref{thm:main} to $C_m$, which produces a subsystem code with 
$n \leq b mn_0^{m+1}$ physical qubits, one logical qubit, distance $d = d_0^m$, and 
$O(1)$-weight gauge generators, where $b$ is the implied constant from 
Theorem~\ref{thm:main}. We would like a bound of the form $d \ge n^{1-\eps}$ for $0<\eps$ 
as small as possible. Taking the log of this inequality gives
\begin{align}
	\frac{\log n -\log d}{\log n} \le \frac{\log n+m \log(b m) + m (1+m) \log\frac{1}{\delta}}{m \log n} \le \eps\,.
\end{align}
The best upper bound is obtained by choosing $m = O\Bigl(\sqrt{\frac{\log n}{\log 1/\delta}}\Bigr)$, 
and a straightforward calculation yields 
\begin{align}
	\eps = 2 \sqrt{\frac{\log \frac{1}{\delta}}{\log n}} + O\Bigl(\frac{\log \log n}{\log n}\Bigr) \,.
\end{align}

A slightly improved constant and stronger control of the subleading term can be obtained using 
homological product codes~\cite{BH13} as the base code for the construction. This lowers the 
total weight of the concatenated generators slightly to $\le O(n_0^{m+1/2})$. We can improve our estimate of $\eps$ to
\begin{align}
	\eps = \sqrt{\frac{2 \log \frac{1}{\delta}}{\log n}} + O\Bigl(\frac{1}{\log n}\Bigr) \,
\end{align}
by tracing through the previous argument.
\end{proof}

\section{Making Sparse Codes Local}\label{S:local}

We can use SWAP gates, identity gates, and some rearrangement to embed the
circuits from Theorem~\ref{thm:main} into $D$ spacetime dimensions so that all
gates become spatially local.
The codes constructed in this way are not just sparse, but also geometrically
local. This results in nearly optimal distances of
$\Omega\bigl(n^{1-1/D-\epsilon}\bigr)$ \cite{Bravyi2009}, as well as spatially
local codes that achieve $kd^{2/(D-1)} \ge \Omega(n)$ in $D=4$ dimensions.

We will lay out our circuits in $D-1$ spatial
dimensions and $1$ time dimension. The \emph{depth} of a circuit refers to its
length in the time dimension, and its \emph{cross section} is the volume of the smallest
$D-1$-dimensional spatial bounding box such that the bounding box contains 
every time slice of the circuit. 

The first step will be to establish a way to take sparse circuits
that consist of mutually commuting subcircuits and compactly arrange them into $D$ spacetime dimensions.
This
then makes it possible to show that fault-tolerant measurement gadgets
can be embedded compactly, and finally to construct error-detecting circuits
that lead to spatially local codes. We begin by extending our notion of
sparsity from codes to circuits and subcircuits.

\begin{definition}
Let $M$ be a set of mutually commuting subcircuits that act on some set of
data wires. $M$ is \emph{$k$-sparse} for some $k$ if each subcircuit in $M$
acts on at most $k$ data wires, and each data wire is acted on by at most $k$
subcircuits of $M$.
\end{definition}

\begin{lemma}\label{lem:permutations}
Suppose there are $n$ data wires and a $k$-sparse set $M$ of commuting subcircuits 
that each have depth at most $h$ and cross section $O(1)$ when embedded in $D-1$
spatial dimensions. Then there is a circuit that is local in $D-1$ spatial
dimensions that enacts all subcircuits in $M$, using depth
$O\bigl(h+n^{1/(D-1)}\bigr)$ and cross section $O(n)$, treating $D$ and $k$ as
constants.
\end{lemma}
\begin{proof}
First we split the commuting subcircuits into $O(k^2)$ layers such that each
layer can be enacted in parallel on disjoint data wires. Consider the
graph with vertex set $M$ and with an edge between each pair of subcircuits
that share some common data wire that they both act on. Since $M$ is sparse,
the maximum degree in this graph is $O(k^2)$. A greedy coloring of this graph 
uses $O(k^2)$ colors, so we simply use each color as a separate layer.

We arrange the data wires arbitrarily in a hypercube lattice of side length
$O(n^{1/(D-1)})$ so that the cross section of the circuit is $O(n)$.
Before each layer of subcircuits, it will be necessary to permute the data
wires to be within reach of the subcircuits that act on them. Such permutation
circuits exist, using only nearest-neighbor SWAP gates, with depth at most
$O\bigl(n^{1/(D-1)}\bigr)$ \cite{Thompson1977}.

Hence the circuit contains $k^2$ layers, each of which consists of a
permutation step with depth $O\bigl(n^{1/(D-1)}\bigr)$ and a
subcircuit-enacting step with depth $h$.
\end{proof}

\begin{lemma}\label{lem:localftmeas}
A fault-tolerant measurement gadget for weight $w$ can be embedded in
$D-1$ spatial dimensions with depth $O\bigl(w^{1/(D-1)}\bigr)$ and cross section $O(w)$.
\end{lemma}
\begin{proof}
Recall that the sparse version of the gadgets consists of a data block, a cat
block, and a parity block (Section~\ref{sec:ft-gadgets} and Fig.~\ref{fig:ftmeas}). 
To achieve locality, all of these blocks, which contain a total of $O(w)$ wires 
including pre- and postselections, can be arranged arbitrarily in hypercubic 
lattices of side length $O\bigl(w^{1/(D-1)}\bigr)$ and 
then spatially interlaced so that each data wire is next to a corresponding
cat wire and a few parity wires. Overall the cross section is still $O(w)$.


The multi-target CNOTs in the gadgets all commute with
each other, so Lemma $\ref{lem:permutations}$ applies, with each gate as a
separate subcircuit. This again takes only $O\bigl(w^{1/(D-1)}\bigr)$ time
steps, so that the overall depth is $O\bigl(w^{1/(D-1)}\bigr)$.
\end{proof}

\begin{theorem}[Restatement of Thm.~\ref{thm:localcodes}]\label{thm:localcodes2}
Spatially local subsystem codes exist in $D\ge 2$ dimensions with gauge generators of weight 
$O(1)$ and minimum distance $d = n^{1-\eps-1/D}$ where $\eps = O(1/\sqrt{\log{n}})$.
\end{theorem}

\newcommand{\tth}{^\text{th}}
\begin{proof}
Following the outline of the proof of Theorem~\ref{thm:catcodes}, we start with a base 
code $C_0$ with parameters $[n_0, 1, d_0]$ and concatenate it $m$ times to obtain 
$C_m$, a code with parameters $[n_0^m, 1, d_0^m]$. The $j\tth$ level of
concatenation will consist of $n_0^{m-j}$ different ``cells'' of qubits, each
of which supports precisely one copy of the base code at the appropriate level of
concatentation. Each cell exposes one
virtual qubit to the next level of concatenation, and at the same time is
composed of $n_0$ different virtual qubits 
from the previous level. We arrange the base code as a $(D-1)$-dimensional
cube of side length $O\bigl(n_0^{1/(D-1)}\bigr)$. Each cell in the
concatenation then uses that same spatial layout, but with each qubit of the
base code replaced by a cell from the previous level.
Thus the linear length scale increases by a factor of $O\bigl(n_0^{1/(D-1)}\bigr)$ at
each successive level of concatenation, such that the overall side length is
$O\bigl(n^{1/(D-1)}\bigr)$ and the overall cross section is $O(n)$. This
creates a recursive, self-similar pattern that resembles a Sudoku puzzle
layout in the example of $D=2$ spatial dimensions.

We arrange the stabilizer postselections into $m$ layers through the time
dimension, with each layer corresponding to a level of concatenation.
Within the $j\tth$ layer, each of the $j\tth$-level cells postselects their
stabilizers at the same time; this can be done in parallel since all of the
stabilizers are contained within their own cells.
There are at most $n_0$ stabilizers to be measured per cell, each of which have
weight at most $n_0^j$. By Lemma \ref{lem:localftmeas}, measuring these will
use depth $O(n_0^{j/(D-1)})$ while keeping the cross section at $O(n_0^j)$ per cell, or
$O(n)$ total.

Summing over all the layers of concatenation, the dominant contribution is from
the final $m$th layer, with a volume of $O\bigl(n^{1+1/(D-1)}\bigr)$. We can
repeat the argument from Theorem~\ref{thm:catcodes}, and choosing
$m = O(\sqrt{\log n})$ we find the stated scaling claim. 
\end{proof}

This same idea can be used to take existing codes that are sparse but not embeddable 
in low-dimensional spaces without large distortion and turn them into subsystem 
codes that are still sparse, but spatially local in $D$ dimensions for constant $D$. 
The next theorem does exactly this, and with some interesting consequences.

For the class of \emph{local commuting projector codes} in $D$ dimensions, 
Bravyi, Poulin, and Terhal have proven~\cite{Bravyi2010a} the following inequality
\begin{align}\label{eq:bptbound}
	k d^{2/(D-1)} \le O(n) \,.
\end{align}
This class of codes is more general than stabilizer codes, but does not include subsystem 
codes, so this bound does not apply to our construction. Nonetheless, it is interesting for 
comparison. The next theorem shows that any analogous bound for subsystem codes cannot 
improve beyond constant factors, and also that improving the existing bound would require 
using structure that isn't present in general subsystem codes.

\begin{theorem}\label{thm:kdbound}
For each $D \ge 2$ there exist spatially local sparse subsystem codes with parameters
$k \ge \Omega\bigl(n^{1-1/D}\bigr)$ and $d \ge \Omega\bigl(n^{(1-1/D)/2}\bigr)$.
In particular, $k d^{2/(D-1)} \ge \Omega(n)$.
\end{theorem}
\begin{proof}
We start with the $[n_0, \Omega(n_0), \Omega(\sqrt{n_0})]$ quantum LDPC code 
due to Tillich and Z\'{e}mor~\cite{TZ09}. This is a sparse stabilizer code, so in its 
syndrome-detecting circuit the stabilizer-postselection gadgets form a sparse set of
mutually commuting subcircuits that are each of bounded size. Hence 
Lemma~\ref{lem:permutations} applies, and yields a spatially local 
Clifford circuit with size $O\bigl(n_0^{1+1/(D-1)}\bigr)$.

Converting this circuit into a code then gives us the parameters
$[n, \Omega\bigl(n^{1-1/D}\bigr), \Omega\bigl(n^{(1-1/D)/2}\bigr)]$ in a
$D$-dimensional spatially local code.
\end{proof}


\section{Discussion}\label{S:conclusion}

The construction presented here leaves numerous open questions.

We have not addressed the important issue of efficient decoders for these codes. It seems 
likely that the subsystem code can be decoded efficiently if the base code can, but we have 
not yet checked this in detail and leave this to future work. One potential stumbling block is 
that the subsystem code requires measuring gauge generators which must be multiplied 
together to extract syndrome bits. Since the stabilizers for our subsystem codes are in general 
highly nonlocal and products of many gauge generators, this might lead to difficulties in 
achieving a fault-tolerant decoder in the realistic case of noisy measurements.

Another open question is whether the distance scaling of Theorems~\ref{thm:catcodes} 
and~\ref{thm:localcodes} can be extended to apply also to $k$ to some degree. 
Improving the fault-tolerant gadgets or using specially designed base 
codes seem like obvious avenues to try to improve on our codes. 
Conversely, extending the existing upper 
bounds by Bravyi~\cite{Bravyi2011b} to $D>2$, as well as extending the bound from 
Bravyi, Poulin, and Terhal~\cite{Bravyi2010a} to subsystem codes would be also 
be interesting.  If the bound in Eq.~(\ref{eq:bptbound}) extended to subsystem codes, 
then the scaling in Theorem~\ref{thm:kdbound} would be tight; alternatively, it is
possible that spatially local subsystem codes can achieve asymptotically better parameters
than spatially local stabilizer codes.

It seems plausible that the base codes for our construction could be extended to 
include subsystem codes. However, it is not known if this would imply any improved code parameters.

It is still an open question whether any distance greater than $O(\sqrt{n\log n})$ can be 
achieved for stabilizer codes with constant-weight generators. If an upper bound on the 
distance for such stabilizer codes were known, then it could imply an asymptotic separation 
between the best distance possible with stabilizer and subsystem codes with constant-weight 
generators, like the separation for spatially local codes in $D=2$ dimensions~\cite{Bravyi2011b}.

Another open question is whether the recent methods of Gottesman~\cite{Gottesman13} for 
using sparse codes in fault-tolerant quantum computing (FTQC) schemes can be modified to 
work with subsystem codes. If they could, then improving our scaling with $k$ would imply that 
FTQC is possible against adversarial noise at rate $R =\exp(-c\sqrt{\log n})$.
Even if not, subsystem codes have in the past proven useful for
FTQC~\cite{AC07} and we are hopeful our codes might assist in further developments of
FTQC.

Finally, it is tempting to speculate on the ramifications of these codes for 
self-correcting quantum memories. The local versions of our codes in 3D have no string-like 
logical operators. To take advantage of this for self-correction, we need a local Hamiltonian 
that has the code space as the (at least quasi-) degenerate ground space and a favorable 
spectrum. 
The underlying code should also have a threshold against random errors~\cite{Pastawski2009}.
The obvious choice of Hamiltonian is minus the sum of the gauge generators, but 
this will not be gapped in general. Indeed, the simplest example of a Clifford circuit -- a wire of 
identity gates -- maps directly onto the quantum $XY$ model, which is gapless when the 
coupling strengths are equal~\cite{Lieb1961}, but somewhat encouragingly is otherwise 
gapped and maps onto Kitaev's proposal for a quantum wire~\cite{Kitaev2001}. Other models 
of subsystem code Hamiltonians exist; some are 
gapped~\cite{Brell2011, Bravyi2012a, Brell2014} and some are 
not~\cite{Bacon2006, Dorier2005}. Addressing the lack of a satisfying general theory of gauge 
Hamiltonians is perhaps a natural first step in trying to understand the power of our 
construction in the quest for a self-correcting memory. 

\acknowledgments

We thank Sergey Bravyi for suggesting the argument in Theorem~\ref{thm:catcodes2}, 
Larry Guth for explanations about \cite{GL13}, Jeongwan Haah for
explaining the distance bounds for the cubic code, David Poulin for discussions and 
to Andrew Doherty, David Long, and an 
anonymous reviewer who found serious errors in a previous version.
DB was supported by the NSF under Grants No.~0803478, 0829937, and 0916400 and 
by the DARPA-MTO QuEST program through a grant from AFOSR. 
STF was supported by the IARPA MQCO program, 
by the ARC via EQuS project number CE11001013, 
by the US Army Research Office grant numbers W911NF-14-1-0098 and W911NF-14-1-0103, 
and by an ARC Future Fellowship FT130101744. 
AWH was funded by NSF grants CCF-1111382 and CCF-1452616 and ARO contract
W911NF-12-1-0486.
JS was supported by the Mary Gates Endowment, Cornell University Fellowship,
and David Steurer's NSF CAREER award 1350196.

\end{document}